\documentclass[a4paper,twocolumn,11pt, accepted=2023-08-29]{quantumarticle}
\pdfoutput=1
\usepackage[utf8]{inputenc}
\usepackage[english]{babel}
\usepackage[T1]{fontenc}
\usepackage{amsmath}
\usepackage{hyperref}
\usepackage{tikz}
\usepackage{lipsum}

\usepackage[numbers, sort&compress]{natbib}
\usepackage{array}
\usepackage{graphicx}
\usepackage{mathdots}
\usepackage{amsfonts,amsthm,amssymb,bbm, fullpage,braket,float}
\usepackage[ruled,linesnumbered]{algorithm2e}
\usepackage{multirow}
\usepackage{makecell}
\usepackage[normalem]{ulem}
\usepackage{yhmath}
\usetikzlibrary{arrows}

\newtheorem{theorem}{Theorem}
\newtheorem{prop}{Proposition}
\newtheorem{lemma}[theorem]{Lemma}
\newtheorem{obs}[theorem]{Observation}
\newtheorem{corollary}[theorem]{Corollary}
\newtheorem{conjecture}[theorem]{Conjecture}

\newcommand{\dimension}{{\color{black}dimension }}
\newcommand{\dimensional}{{\color{black}dimensional }}

\newcommand{\R}{\mathbb{R}}
\newcommand{\E}{\mathbb{E}}
\newcommand{\spn}{\mathrm{span}}
\newcommand{\mc}{{\sc Max-Cut}}

\newcolumntype{?}{!{\vrule width 1pt}}

\begin{document}

\title{Warm-Started QAOA with Custom Mixers Provably Converges and Computationally Beats Goemans-Williamson's Max-Cut at Low Circuit Depths}
\author{Reuben Tate}
\affiliation{CCS-3 Information Sciences, Los Alamos National Laboratory, Los Alamos, NM 87544, USA}
\orcid{0000-0002-9170-8906}
\author{Jai Moondra}
\orcid{0000-0002-6401-1505}
\affiliation{Georgia Institute of Technology, Atlanta, GA 30332, USA}
\author{Bryan Gard}
\affiliation{Georgia Tech Research Institute, Atlanta, GA 30332, USA}
\orcid{0000-0001-5529-1675}
\author{Greg Mohler}
\affiliation{Georgia Tech Research Institute, Atlanta, GA 30332, USA}
\orcid{0000-0001-7330-5962}
\author{Swati Gupta}
\affiliation{Sloan School of Management, Massachusetts Institute of Technology, Cambridge, MA 02142, USA}
\orcid{0000-0002-9566-3856}
\thanks{email of the corresponding author is \href{mailto:swatig@mit.edu}{swatig@mit.edu}}
\maketitle

\begin{abstract}
We generalize the Quantum Approximate Optimization Algorithm (QAOA) of Farhi et al. (2014) to allow for arbitrary separable initial states with corresponding mixers such that the starting state is the most excited state of the mixing Hamiltonian. We demonstrate this version of QAOA, which we call {\it QAOA-warmest}, by simulating Max-Cut on weighted graphs. We initialize the starting state as a {\it warm-start} using $2$ and $3$-\dimensional  approximations obtained using randomized projections of solutions to Max-Cut's semi-definite program, and define a warm-start dependent \textit{custom mixer}. We show that these warm-starts initialize the QAOA circuit with constant-factor approximations of 0.658 for 2-\dimensional and 0.585 for 3-\dimensional warm-starts for graphs with non-negative edge weights, improving upon previously known trivial (i.e., 0.5 for standard initialization) worst-case bounds at $p=0$. These factors in fact lower bound the approximation achieved for Max-Cut at higher circuit depths, since we also show that QAOA-warmest with any separable initial state converges to Max-Cut under the adiabatic limit as $p\rightarrow \infty$. However, the choice of warm-starts significantly impacts the rate of convergence to Max-Cut, and we show empirically that our warm-starts achieve a faster convergence compared to existing approaches. Additionally, our numerical simulations show higher quality cuts compared to standard QAOA, the classical Goemans-Williamson algorithm, and a warm-started QAOA without custom mixers for an instance library of 1148 graphs (upto 11 nodes) and depth $p=8$. We further show that QAOA-warmest outperforms the standard QAOA of Farhi et al. in experiments on current IBM-Q and Quantinuum hardware.
\end{abstract}

\section{Introduction}
\label{sec:introduction}

In order to realize a quantum advantage, many researchers have been considering the usage of NISQ devices \cite{P18,HM2017} for the purposes of solving difficult problems in combinatorial optimization. Of particular interest is the Quantum Approximate Optimization Algorithm (QAOA), a hybrid quantum-classical algorithm developed by Farhi et al. that can be applied to a large variety of combinatorial optimization problems \cite{FGG14}. We study the use of QAOA to solve one of the most famous NP-hard combinatorial optimization problems, called Max-Cut.\footnote{Although this work focuses on Max-Cut, our approach can be applied to any suitable combinatorial optimization problem by converting a generic QUBO instance into a Max-Cut instance \cite{DGS18}, and then applying the same techniques. This conversion requires only one additional qubit; however, such a conversion is not necessarily approximation-preserving.} Given a weighted graph $G = (V,E)$, with vertex set $V = [n]$, edge set $E \subseteq \binom{V}{2}$ and weights $w: E \to \mathbb{R}$, the Max-Cut problem is to find a partition of $V$ into two disjoint sets $S, V \setminus S \subseteq V$, such that the total weight of the edges across the partition, i.e. $\text{cut}(S) := \sum_{e \in E} w_e \cdot \mathbf{1}[e \in S \times (V \setminus S)]$, is maximized. The Max-Cut of $G$ is denoted by $\text{Max-Cut}(G) = \max_{S \subseteq V} \text{cut}(S)$.

In the standard QAOA algorithm, qubits are initialized in the $\ket{+}$ state along the $x$-axis of the Bloch sphere, tensorized $n$ times, and QAOA's mixing Hamiltonian rotates each qubit about this same axis. The promise of this approach lies in Farhi et al.'s \cite{FGG14} result that establishes a connection between standard QAOA and quantum adiabatic computing, showing that with increased circuit depth, standard QAOA converges to the Max-Cut. This connection relies on the initial state of standard QAOA being  the highest energy eigenstate of the mixing Hamiltonian. 


We propose to initialize the circuit with separable initial states (other than $\ket{+}^{\otimes n}$) generated using classical relaxations of the Max-Cut problem; following classical optimization literature, we refer to these states as {\it warm-starts}. We modify the mixing Hamiltonians in the QAOA ansatz in a way dependent on the initial state, so that the adiabatic assumptions hold; we call these mixing Hamiltonians {\it custom mixers}, and call the new QAOA variant {\it QAOA-warmest}. We show that like QAOA, QAOA-warmest converges to the Max-Cut with increased circuit depth (Proposition \ref{thm:convergenceGeneral}).

We further note that the convergence \emph{rate} of QAOA-warmest is heavily dependent on the separable initial state used. In this work, we focus on two approaches for warm-starting that generate classically-inspired separable initial states which (empirically) converge faster compared to existing initializations: (1) low-\dimensional projections of optimal solutions to the Goemans-Williamson (GW) semidefinite program (SDP) \cite{GW95}, and (2) locally optimal solutions to  low-\dimensional Burer-Monteiro \cite{BM03} relaxations of the Max-Cut problem. Using a diverse library of 1148 graphs (up to 11 nodes), we show through numerical simulations that QAOA-warmest that uses our two types of warm-starts outperforms the classical Goemans-Williamson algorithm and standard QAOA even at low-circuit depths around $p = 4$. Such a low-depth for this performance does not hold for any random initialization on the Bloch sphere. We further tested QAOA-warmest in the presence of noise using Qiskit's built-in modules and hardware calibration data \cite{qiskit}, and on IBM's 16-qubit Guadalupe device and Quantinuum's 20-qubit devices. Our simulations demonstrate that QAOA-warmest is robust and still yields high (instance-specific) approximation ratios in such noisy regimes.

Additionally, we provide theoretical performance guarantees at depth $p=0$ in the case that projected GW solutions are used. Specifically, for positive-weighted graphs, we show  (Theorem \ref{thm: subspace-hyperplane-rounding}) that, in expectation, the value of the cut obtained from hyperplane rounding of the projected GW solutions is at least $0.878$ of the optimal cut value (just like GW). We further show that quantum measurement of the projected solutions (i.e. depth-0 QAOA-warmest) yields at least a $\frac{3}{4}(0.878) \approx 0.658$-approximation and $\frac{2}{3}(0.878) \approx 0.585$-approximation for dimensions $k=2$ and $k=3$ respectively. These guarantees also serve as a lower bound for what can be achieved with higher depth since, like standard QAOA, our approach also has monotonically increasing expected cut-quality with increased circuit depth for any instance.

This paper presents the first warm-start approach for QAOA (for Max-Cut) which simultaneously (1) provides a nontrivial constant-factor approximation ratio at depth $p=0$, (2) satisfies provable convergence to Max-Cut as circuit depth increases, and (3) enjoys a fast convergence rate, as shown through numerical simulations and experiments on quantum hardware. Some of the previous approaches \cite{CFGRT22, egger2020warm} consider warm-start initializations using perturbations of a single-cut (obtained for example, using GW algorithm). While such initializations theoretically yield higher approximation ratio of 0.878 using quantum sampling at depth $p=0$, this warm-start has a much slower convergence rate in simulations (see Table \ref{tab:summaryTable}). Therefore, we believe that our approach is promising for near-term NISQ devices.


\subsection{Related Work}
Since Farhi et al.'s \cite{FGG14} seminal paper in 2014, many have researched and analyzed the QAOA algorithm. Many empirical studies have been performed including the effects of different parameter initialization strategies \cite{SMKS22, SS21, ZWCPL20, SLLOH22, GLLAS21}, the performance of QAOA in the context of devices with superconducting qubits with limited connectivity \cite{WVGTBWE22,lotshaw2022scaling,guerreschi2019qaoa},the use of machine learning to determine (for specific instances) whether Max-Cut QAOA would yield better cuts compared to the classical GW algorithm \cite{moussa2020quantum}, the effectiveness of different encoding schemes for objectives involving higher-order terms with more than 2 qubits \cite{campbell2022qaoa}, the effect of various graph properties on the performance of QAOA \cite{HTOLHS21}, and the use of QAOA to generate highly squeezed states which are useful in the context of quantum metrology \cite{SJHE22}.

Many researchers have also analyzed QAOA from a more theoretical perspective. Shaydulin et al. \cite{shaydulin2021classical} gives a series of results regarding the classical symmetries in the objective function and how those symmetries are reflected in the QAOA dynamics. In their seminal paper, Farhi et al. show that depth-1 QAOA achieves an approximation ratio of 0.694 for Max-Cut on 3-regular graphs \cite{FGG14}. Wurst and Love show that at depth-2, this approximation ratio improves to 0.7559 for Max-Cut QAOA on 3-regular graphs \cite{wurtz2021maxcut}.  For quantum devices with limited connectivity, Farhi et al. \cite{FGGN17} show that a variation of Max-Cut QAOA on 3-regular graphs on a device whose native graph is a square grid achieves an approximation ratio of 0.5293 without the use of swap operations. Others have also proven limitations of the standard QAOA algorithm as well: Bravyi et al. \cite{BKKT19} show that, for all $d\geq 3$, there exists a sequence of $d$-regular bipartite graphs such that depth-$p$ QAOA with $p < (1/3 \log_2 n - 4)d^{-1}$ on such instances produces a cut (in expectation) whose value is at most $\frac{5}{6}+\frac{\sqrt{d-1}}{3d}$, meaning that, in the worst case, constant-depth QAOA for Max-Cut is inferior to the classical Goemans-Williamson algorithm as $\lim_{d\to \infty}\left(\frac{5}{6}+\frac{\sqrt{d-1}}{3d}\right) = \frac{5}{6} \approx 0.833 < 0.878$. Farhi et al. \cite{FGG20} show a similar result when QAOA is applied to the Max Independent Set problem\footnote{Given a graph $G=(V,E)$, the goal of the Max Independent Set problem is to find $S \subseteq V$, with $|S|$ as large as possible, such that for all vertices $u,v \in S$, the edge $(u,v) \in E$.} and Bravyi et al. \cite{bravyi2022hybrid} give similar results for a recursive variant of QAOA applied to the Max-$k$-Cut problem.\footnote{In the Max-$k$-Cut problem, one is given a weighted graph $G=(V,E)$ with weights $w: E \to \mathbb{R}$ and the goal is to partition the vertices into $k$ disjoint groups so that the sum of weights of edges across partitions is maximized} Hastings \cite{hastings2019classical} defines a notion of a ``local" classical algorithm and shows that, for triangle-free $d$-regular graphs with $2 \leq d \leq 1000$, there exists\footnote{Numerical evidence suggests that the results of Hastings \cite{hastings2019classical} and Marwaha \cite{marwaha2021local} hold for all $d$; however, no formal proof is given.} a $d$-dependent local classical algorithm for Max-Cut with a provably better approximation ratio compared to depth-1 QAOA. Marwaha \cite{marwaha2021local} extended this result, showing that for each $2 \leq d \leq 500$, there exists local classical algorithms that yields a better expected cut compared to depth-2 QAOA for all $d$-regular graphs with girth greater than 5. Barak and Marwaha \cite{barak2021classical} have continued this line of research, showing that for every one-local algorithm (classical or quantum), that the maximum cut achieved is at most $\frac{1}{2}+\frac{1/\sqrt{2}}{d}$ of the maximum cut for $d$-regular graphs with girth greater than 5 and that there exists a $k$-local algorithm with approximation ratio $\frac{1}{2}+\frac{2/\pi}{d} - O(\frac{1}{\sqrt{k}})$ for $d$-regular graphs with girth greater than $2k+1$. However, general approximation guarantees obtained by the QAOA algorithm still remain elusive.

The above results also suggest that in order to realize some kind of quantum advantage, more information beyond the standard QAOA algorithm might be needed. Recent work has considered modifications and variations of the QAOA algorithm itself, as in this work, which we discuss next. 

The closest works related to this paper are those by Tate et al. \cite{TFHMG20} and Egger et al. \cite{egger2020warm}, where the authors have explored multiple approaches for warm-starting QAOA. Tate et al. \cite{TFHMG20} considered warm-starting QAOA using $2$ and $3$-\dimensional Burer-Monteiro locally optimal solutions for Max-Cut; however, their method plateaus even at low circuit depths of $p=1$ for some initializations, and it is unable to improve the cut quality for some instances. For the Max-Cut problem, Egger et al. \cite{egger2020warm} constructed the initial quantum state by (non-trivially) mapping a single specific cut $(S, V \setminus S)$ (obtained via the Goemans-Williamson algorithm or possibly other means) to an initial quantum state. Egger et al. also modified the mixing Hamiltonian so that at depth $p=1$ (with the right choice of QAOA variational parameters), the cut $(S, V \setminus S)$ is recovered; however, there is no evidence to suggest that such an approach will converge to the optimal solution for depth $p \to \infty$. They further proposed a different continuous warm-started QAOA for Quadratic Unconstrained Binary Optimization (QUBO) problems, which enjoys convergence guarantees. For certain classes of QUBO's, the binary variables can be relaxed to be in the interval $[0,1]$ to obtain a convex quadratic program whose optimal solution can then be mapped to the Bloch sphere. Our mixer is a generalization of the mixer used in this particular approach by Egger et al.; however, the initialization scheme is quite different. In particular, their continuous warm-started QAOA approach cannot be directly applied to Max-Cut as the corresponding relaxed quadratic program is not convex.

Recent work by Cain et al. \cite{CFGRT22} further explored convergence properties of warm-starts, when augmented with standard mixers. They showed using a perturbative approach that standard QAOA with single-cut warm-starts (i.e., each qubit is initialized at $\ket{0}$ or $\ket{1}$) does not show any improvement in the approximation ratio, even when the circuit depth is increased. This result is interesting in the context of our work, since we find that (1) using custom mixers, one can guarantee convergence as long as the initialization is not at a single-cut, and (2) warm-starts that are perturbations of single-cut initializations (e.g., each qubit is initialized at $R_Y(\theta^*)\ket{0}$ or $R_Y(\pi-\theta^*)$ for some small $\theta^*$ where $R_Y(\theta)$ is a single-qubit rotation about the $y$-axis by angle $\theta$) converge very slowly, in computations, even with custom mixers. Our work therefore complements the existing work and clarifies what kind of warm-starts are useful. 

Additionally, other works have explored different kind of modifications of the QAOA ansatz. Farhi et al. \cite{FGGN17} considered having separate variational parameters for each vertex and edge. Hadfield et al. \cite{HWORVB17} and Wang et al. \cite{WRDR20} considered versions of QAOA that are suitable for combinatorial optimization problems with both hard constraints (that must be satisfied) and soft constraints (for which we want to minimize violations). Zhu et al. \cite{ZTBMBE20} modify QAOA such that the ansatz is expanded in an iterative fashion with the mixing Hamiltonian being allowed to change between iterations. Bravyi et al. \cite{BKKT19} proposed a recursive QAOA approach that decreases the instance size at each iteration; Egger et al. \cite{egger2020warm} also consider a similar recursive version of their approaches. B\"artschi et al. \cite{bartschi2020grover} and Jiang et al. \cite{jiang2017near} consider modifications of QAOA inspired by Grover's (quantum) algorithm \cite{grover1996fast} for fast database search. For the scope of this work, we do not consider these alternate approaches; however, it may serve as an interesting direction for future work as QAOA-warmest can likely be used in conjunction with these other approaches.

\subsection{Background}

\label{sec:background}
\subsubsection{The Quantum Approximate Optimization Algorithm}
\label{sec:QAOA}
First, we review QAOA in the context of the Max-Cut problem. QAOA is performed on $n$ qubits (with $n$ being the number of vertices in the input graph) with the final measurement being a bitstring of length $n$ corresponding to a cut $(S,V \setminus S)$ with $S = \{i : i\text{th bit is } 0\}$. 

For depth-$p$ QAOA, the quantum circuit alternates between applying a cost Hamiltonian $H_C = \frac{1}{2} \sum_{(i,j)\in E} w_{ij} (1 - \sigma_i^z \sigma_j^z)$ and a mixing Hamiltonian $H_B = \sum_{i \in [n]} \sigma_i^x$. Here, $\sigma^x,\sigma^y,\sigma^z$ are the standard Pauli operators and $\sigma_i^k$ is the Pauli operator applied to qubit $i$ for $k\in \{x,y,z\}$. The cost and mixing Hamiltonians are applied a total of $p$ times to generate a variational state,
\begin{align}
&\ket{\psi_p(\gamma, \beta)} \label{eqn:qaoaCircuit}\\
& =  e^{-i\beta_p H_B} e^{-i\gamma_p H_C} \! \cdots e^{-i\beta_1 H_B} e^{-i\gamma_1 H_C} \! \ket{ s_0 }, \notag
\end{align}
where $\ket{s_0}$ is the initial state and $\gamma = (\gamma_1,\dots,\gamma_p), \beta = (\beta_1,\dots,\beta_p)$ are variational parameters to be optimized. For standard QAOA, the initial state is given by $\ket{s_0} = \ket{+}^{\otimes n}$ which corresponds to an equal superposition of all $2^n$ possible cuts in the graph.

Finally, sampling from $\ket{\psi_p(\gamma, \beta)}$ yields a bit-string corresponding to a cut in the graph. We let $F_p(\gamma,\beta)$ denote the expected cut value from sampling $\ket{\psi_p(\gamma, \beta)}$, i.e.,
$$F_p(\gamma,\beta) = \bra{ \psi_p(\gamma,\beta) } H_C \ket{\psi_p(\gamma, \beta)}.$$
If, for each $p$, we choose $\gamma,\beta$ optimally, then the expected value of the cut tends to the Max-Cut as $p \to \infty$ \cite{FGG14}.

In practice, the optimal choice of $\gamma,\beta$ is unknown in advance and thus a hybrid classical-quantum hybrid loop is often used to find values of $\gamma$ and $\beta$ that yield high expectation values. The initialization of $\gamma$ and $\beta$ in these optimization routines has been investigated \cite{SMKS22, SS21, ZWCPL20, SLLOH22, GLLAS21}; however, for simplicity, we initialize $\gamma$ and $\beta$ near the origin for our experiments as discussed in Section \ref{sec:experiments}. 

\subsubsection{Classical Optimization Algorithms}
\label{subsec:classicalAlgoPrelims}
We next review some classical optimization algorithms for Max-Cut. It is useful to reformulate the Max-Cut problem  as follows:
\begin{equation}\text{Max-Cut}(G) = \frac{1}{4}\max_{x \in \{\pm 1\}^n} \sum_{(i,j) \in E}w_{ij}(x_i - x_j)^2. \label{eqn:maxcutFormulation}\end{equation}
Given a solution $x$ to the maximization above, this yields the cut $(S,V \setminus S)$ with $S =\{i : x_i = 1\}$. For positive integer $k$, the above formulation can then be relaxed as follows:

\begin{align}
	\text{max} & \quad \frac{1}{4}\sum_{(i,j) \in E}w_{ij}\Vert x_i - x_j\Vert^2  \label{eqn:rank-k-relaxation}\\
	\text{subject to} &\quad  \Vert x_i\Vert = 1,&& \hspace{-2cm} \forall i \in V, \nonumber \\
	&\quad x_i \in \mathbb{R}^{k},&& \hspace{-2cm} \forall i \in V.\nonumber
\end{align}

When $k=n$, the problem can be reformulated as a convex semidefinite program (SDP) which is the formulation used by the seminal Goemans-Williamson (GW) algorithm \cite{GW95}, which yields a 0.878-approximation to Max-Cut for graphs with non-negative edge weights.

Burer and Monteiro proposed solving the relaxation for $1 < k < n$ (which we denote by BM-MC$_k$), using parametric forms for points on the $k-1$-dimensional sphere \cite{BM03}. This leads to a program that is no longer convex and thus one can not expect to easily find the global optima. However, high-quality local optima can be found by utilizing first and second order optimization methods. In general, the Burer-Monteiro technique has been found to work well in practice, even when $k=2$ \cite{BMZ01}; however, in theory, Mei et al. \cite{MMMO17} show that hyperplane rounding of a locally optimal BM-MC$_k$ solution achieves a $0.878(1 - \frac{1}{k-1})$ fraction of the optimal cut, yielding a $0$-approximation and a $0.439$-approximation for dimensions $k=2$ and $k=3$ respectively.

For ease of notation, given a (feasible) BM-MC$_k$ solution $x$, we let $$\text{HP}(x) = \sum_{(i,j) \in E} \frac{1}{\pi}\arccos(x_i \cdot x_j),$$
denote the expected cut value obtained from performing hyperplane rounding on $x$ \cite{GW95} and we let $$\text{BM-MC}_k(x) =  \sum_{(i,j) \in E} \frac{1}{4}\Vert x_i - x_j \Vert^2,$$ denote the BM-MC$_k$ objective at $x$. Lastly, we say that the solution $x$ is $\kappa$-approximate if $\text{HP}(x) \geq \kappa \text{Max-Cut}(G)$ and similarly, $x$ is considered $\kappa$-close if $\text{BM-MC}_k(x) \geq \kappa \text{Max-Cut}(G)$.

\subsubsection{Approximation Ratio}
\label{sec:approxRatio}


For this work, we adopt the same definition and notation for the (instance-specific) approximation ratio (AR) as discussed in Tate et al. \cite{TFHMG20}. Specifically, for a graph $G$ and algorithm $\mathcal{A}$, the (instance-specific) AR is given by,
\begin{equation}\label{eq:approxRatio}\alpha_{\mathcal{A},G} = \frac{\mathcal{E}_{\mathcal{A},G} - \text{{\sc Min-Cut(G)}}}{\text{\mc}(G) -\text{{\sc Min-Cut(G)}}};\end{equation}
where $\mathcal{E}_{\mathcal{A},G}$ is the expected cut value (either from hyperplane rounding of a classical solution or quantum measurement of a quantum state) and $\text{{\sc Min-Cut(G)}}$ is the (possibly trivial) minimum cut in the graph $G$. This definition is chosen as a means of ``normalizing" the approximation ratio to lie in the interval $[0,1]$, even in the case of graphs with both positive and negative edge weights. For positive-weighted graphs, we have that $\text{{\sc Min-Cut(G)}} = 0$, in which case, the definition of $\alpha_{\mathcal{A},G}$ reduces to the ``typical" definition of approximation ratio often seen in the classical optimization literature. The worst-case AR $\alpha_{\mathcal{A}}$ for an algorithm $\mathcal{A}$ is defined as the worst instance-specific AR across all instances, i.e., $\alpha_{\mathcal{A}} = \min_{G} \alpha_{\mathcal{A},G}$; alternatively, we say such an algorithm $\mathcal{A}$ provides an $\alpha_{\mathcal{A}}$-approximation.

It should be noted that in our numerical simulations, we can calculate $\mathcal{E}_{\mathcal{A},G}$ exactly. Please refer to \cite{TFHMG20} for more details.

\section{Custom Mixers}
\label{sec:customMixers}

In this section, we provide a general framework that shows, for \emph{any} separable initial state, how to construct a custom mixing Hamiltonian for QAOA 
so that the warm-start is the most excited state of the mixer. This property will be useful in showing convergence results for QAOA with custom mixers. We refer to such variants of QAOA as QAOA-warmest (compared to QAOA-warm, proposed by Tate et al. \cite{TFHMG20}, which uses the standard mixer, $H_B = \sum_{i \in [n]} \sigma_i^x$, together with a separable quantum initial state). 




Consider a separable state $\ket{s_0}$ on $n$ qubits; if the $j$th qubit's Cartesian coordinates on the Bloch sphere are given by $(x_j, y_j, z_j)$, then the corresponding custom mixing Hamiltonian $H_B$ is given by,
\begin{align}
H_B =\bigoplus_{j=1}^n H_{B,j},\label{eqcustom}
\end{align}
where
$H_{B,j} = x_j\sigma^x + y_j\sigma^y + z_j\sigma^z$. 

Geometrically, the custom mixer rotates qubits about their original position on the Bloch sphere (details included in Appendix \ref{sec:customMixersConstruction}). Note that the  standard mixers in QAOA \cite{FGG14} are, therefore, a special case of our custom mixers since each qubit is initialized at $\ket{+}$, i.e., the the $x$-axis (with Cartesian coordinates $(x_j,y_j,z_j) = (1,0,0)$) and the unitary operator $e^{-i\beta_k H_B}$ for the mixer corresponds to rotations (by $2\beta_k$) about the $x$-axis. When the initial state is composed by qubits restricted to  the $xz$-plane with $x > 0$, then this custom mixer recovers the one considered by Egger et al \cite{egger2020warm}.

\label{sec:descOfCustMixer}

\subsection{Convergence to Max-Cut}\label{sec:convergence-for-separable-initial-states}
In this section, we show that the expected cut obtained by QAOA-warmest converges to Max-Cut as the circuit depth goes to infinity.

\begin{theorem}
\label{thm:convergenceGeneral}
Let $\ket{s_0}$ be any separable initial state whose qubits do not lie at the poles of the Bloch sphere. Running QAOA with a warm-start $\ket{s_0}$, its corresponding custom mixer \eqref{eqcustom}, with the choice of optimal variational parameters, yields a distribution of cuts whose expected value reaches Max-Cut as the circuit depth $p$ tends to infinity, i.e., $$\lim_{p \to \infty} \max_{\gamma,\beta} F_p(\gamma,\beta) = \text{Max-Cut}(G).$$ 
\end{theorem}

We will first show that Theorem \ref{thm:convergenceGeneral} holds when the warm-start is initialized in the $xz$-plane of the Bloch sphere, with $x>0$. We will then show the main result by showing equivalence of the distribution of cuts obtained by QAOA-warmest when a phase is added to each qubit (i.e., the initial state is not restricted to 0 phase).

\begin{prop}
\label{thm:convergenceSpecial}
Let $\ket{s_0}$ be any separable initial state such that all qubits lie at the intersection of the Bloch sphere and the $xz$-plane with positive $x$-coordinate. Running QAOA with initial state $\ket{s_0}$ and its corresponding custom mixer yields that $$\lim_{p \to \infty} \max_{\gamma,\beta} F_p(\gamma,\beta) = \text{Max-Cut}(G),$$ i.e., the expected cut value of QAOA-warmest with optimal variational parameters will yield the optimal cut value as the circuit depth $p$ tends to infinity.
\end{prop}

Egger et al. \cite{egger2020warm} use the same custom mixers\footnote{Although \cite{egger2020warm} use the same construction of the mixer, they use significantly different warm-starts, which ultimately play a role in the rate of convergence of the overall approach.} for warm-starts restricted to the $xz$-plane with $x>0$, for convex quadratic programs, and state that convergence holds, without proof. While straightforward calculations show that the initial state is the highest-energy eigenstate of the mixer (a condition needed in order to apply the adiabatic theorem and guarantee convergence), a complete proof requires careful inspection of the eigenvalues of the time-dependent Hamiltonian $H(t) = (1-t/T)H_B + (t/T)H_C$. We show that there is a non-zero gap between the largest and the second largest eigenvalues of the time-varying Hamiltonian, using the Perron-Frobenius theorem for irreducible stoquastic Hermitian matrices. These details can be found in Appendix \ref{sec:xz_convergence}.



This analysis can not be directly applied for arbitrary separable states since the corresponding mixers do not necessarily have real entries, and therefore, the Perron-Frobenius theorem cannot be applied. However, instead of calculating the eigenvalue gaps directly, we show that the phase in the warm-start is not reflected in the distribution of cuts obtained using QAOA-warmest, which would then complete the proof of Theorem \ref{thm:convergenceGeneral}.


\begin{prop}
\label{thm:convergence}
Let $\ket{s_0}$ be any separable initial state and let $H_B$ be its corresponding custom mixer. Let $\ket{s_0'}$ be the state $\ket{s_0}$ where each qubit's phase is set to 0, and let $H_B'$ be the corresponding custom mixer. Then, for any fixed choice of variational parameters $(\gamma,\beta)$, the distribution of cuts obtained from QAOA-warmest with initial state $\ket{s_0}$ and mixer $H_B$ is the same as the distribution of cuts from obtained QAOA-warmest with initial state $\ket{s_0'}$ and mixer $H_B'$.
\end{prop}

Proposition \ref{thm:convergence} can be proved by first decomposing the general custom mixer (corresponding to $\ket{s_0}$) in terms of the custom mixer corresponding to the initial state with the phase removed (i.e. $\ket{s_0'}$). Using this decomposition, we show that any effect caused by the change in rotation direction (due to switching from the custom mixer for $\ket{s_0'}$ to the custom mixer for $\ket{s_0}$) exactly cancels out with the effect of introducing a phase to $\ket{s_0'}$ to obtain $\ket{s_0}$. The details of the proof can be found in Appendix \ref{sec:proofs}.

The convergence given by Theorem \ref{thm:convergenceGeneral} for QAOA-warmest is especially interesting considering that many previous warm-started QAOA approaches lack such guarantees. The QAOA-warm approach by Tate et al. \cite{TFHMG20} considered arbitrary separable states but with the standard mixer; they showed examples where QAOA-warm plateaued and did not converge to the Max-Cut. Cain et al. \cite{CFGRT22} considered the case where the initialization is a single bitstring/cut with the standard mixer, and showed that such an initialization also does not converge to Max-Cut. One of the approaches by Egger et al. consider a mixer that is neither the standard mixer nor the custom mixer approach presented above. Their mixer is used in conjunction with what we call a perturbed single-cut initialization (see Section \ref{sec:warmstarts}) which recovers a \emph{particular} cut (obtained via GW or other means) at depth 1. However, this initialization comes with no guarantees on convergence, and our results also do not apply to these different mixers. 

Although Theorem \ref{thm:convergenceGeneral} applies to any warm-start with aligned custom mixers, we will show that there is a significant difference in the {\it rate} of convergence to Max-Cut which depends on the type of warm-start chosen. Theoretically characterizing this rate of convergence still remains an open question, but we show that there is a stark observable difference in rate of convergence through numerical simulations (See Table \ref{tab:convergenceAgg} for a summary, and Section \ref{sec:QAOAWarmestCompare} for more details). We discuss the choice of warm-starts next.

\section{Warm-Starts}
\label{sec:warmstarts}

We now discuss the notion of warm-starting the quantum circuit for QAOA by biasing the initial quantum state $\ket{s_0}$ in Equation (\ref{eqn:qaoaCircuit}) to certain cuts, as opposed to taking an equal superposition of all cuts as in standard QAOA. Many initialization schemes have been used for warm-starting QAOA (see Table \ref{tab:summaryTable} for a summary); we focus on two that produce classical solutions which can easily be mapped to a separable quantum state in a way that roughly approximates the corresponding classical distribution of cuts.

Both approaches consider the $k$-dimensional relaxation (\ref{eqn:rank-k-relaxation}) for Max-Cut. When $k = n$, this is the GW semidefinite program. Recall that the Goemans-Williamson algorithm \cite{GW95} rounds an optimal solution $u_1, \ldots, u_n \in \R^{n}$ to this SDP to a cut $v_1, \ldots, v_n \in \{-1, 1\}$ using a random hyperplane. Since we wish to produce a separable quantum state, we instead round $u_1, \ldots, u_n \in \R^{n}$ to vectors $v_1, \ldots, v_n \in \R^{k}$, $k \in \{2, 3\}$ and map these to the Bloch sphere for each of the $n$ qubits.

Alternately, we can solve the relaxation itself for $k \in \{2, 3\}$ \cite{BM03}, and map the locally-optimal solution vectors (called BM-MC solutions) directly to the corresponding Bloch spheres. While BM-MC solutions do not enjoy even the trivial 0.5-approximation guarantees \cite{MMMO17}, they are much faster to compute in practice and therefore scale better for larger graphs: e.g., Burer et al. \cite{BMZ01} show that GW took over 1.5 days to complete on a 20,000 node instance, whereas a $k = 2$ BM-MC solution was found in a little over a second; repeated runs of BM-MC$_2$ over the course of a couple minutes on the same graph yielded cuts that were at least as good as those obtained by GW \cite{BMZ01}. 

\begin{table*}[htbp]
\centering
\resizebox{\linewidth}{!}{%
\begin{tabular}{|l|l|l|l|l|l|}
\hline
Initializations   & Mixers    & Citation        & \begin{tabular}[c]{l}Worst case AR\\ at $p = 0$\end{tabular}                 & \begin{tabular}[c]{l}Converges to\\ Max-Cut?\end{tabular} & Comment                    \\ \hline
\multirow{2}{*}{Equal}    & Standard mixer      & \multirow{2}{*}{\cite{FGG14}}      & \multirow{2}{*}{0.5}            & \multirow{2}{*}{Yes}          & \multirow{2}{*}{\begin{tabular}[c]{l}Standard and\\ custom mixers are\\ equivalent for this\\ initialization\end{tabular}} \\ \cline{2-2}
    & Custom mixer        &                 &                &              &           \\ \hline
\multirow{3}{*}{BM-MC$_k$}                 & \begin{tabular}[c]{l}Standard mixer\\ (QAOA-warm)\end{tabular}   & {\cite{TFHMG20}}      & \begin{tabular}[c]{l}0 $(k = 2)$\\ 0.333 $(k = 3)$\end{tabular}              & No           &           \\ \cline{2-6} 
    & \begin{tabular}[c]{l}{\bf Custom mixers}\\ {\bf (QAOA-warmest)}\end{tabular} & {\bf This work}       & \begin{tabular}[c]{l}{\bf 0 $(k = 2)$}\\ {\bf 0.333 $(k = 3)$}\end{tabular}              & {\bf Yes}          & \begin{tabular}[c]{l}{\bf Converges quickly}\\ {\bf (see Section \ref{sec:QAOAWarmestCompare})}\end{tabular}            \\ \hline
\multirow{4}{*}{{\bf Projected GW}}              & \begin{tabular}[c]{l}{\bf Standard mixer}\\ {\bf(QAOA-warm)}\end{tabular}   & {\bf This work}       & \begin{tabular}[c]{l}{\bf 0.658 $(k = 2)$,}\\ {\bf 0.585 $(k = 3)$,}\\ {\bf (Corollary \ref{thm:preserveHyperplaneRatioQuantum})}\end{tabular} & {\bf No}           &           \\ \cline{2-6} 
    & \begin{tabular}[c]{l}{\bf Custom mixers}\\ {\bf (QAOA-warmest)}\end{tabular} & {\bf This work}       & \begin{tabular}[c]{l}{\bf 0.658 $(k = 2)$,}\\ {\bf 0.585 $(k = 3)$,}\\ {\bf (Corollary \ref{thm:preserveHyperplaneRatioQuantum})}\end{tabular} & {\bf Yes}               & \begin{tabular}[c]{l}{\bf Converges quickly}\\ {\bf (see Section \ref{sec:QAOAWarmestCompare})}\end{tabular}            \\ \hline
\multirow{6}{*}{\begin{tabular}[c]{l}Perturbed \\ Single Cut (with\\ regularization\\ angle $\theta^*$)\end{tabular}} & \begin{tabular}[c]{l}Standard mixer\\ (QAOA-warm)\end{tabular}   & {\cite{CFGRT22}} $(\theta^* = 0)$ case & \begin{tabular}[c]{l}$0.878 \cos^{2n}(\theta^*/2)$,\\ (Observation \ref{thm:singleCutAR})\end{tabular}   & No           &           \\ \cline{2-6} 
    & \begin{tabular}[c]{l}Custom mixers\\ (QAOA-warmest)\end{tabular} & Appendix I of {\cite{egger2020warm}}         & \begin{tabular}[c]{l}$0.878 \cos^{2n}(\theta^*/2)$,\\ (Observation \ref{thm:singleCutAR})\end{tabular}   & Yes          & \begin{tabular}[c]{l}Converges slowly for\\ small $\theta^*$\\ (see Section \ref{sec:QAOAWarmestCompare})\end{tabular}       \\ \cline{2-6} 
    & Other custom mixers                  & Section 2.3 of {\cite{egger2020warm}}       & 0.878          & No           & \begin{tabular}[c]{l}AR result occurs at \\ $(\gamma_1, \beta_1) = (0, \pi/2)$\end{tabular}               \\ \hline
\end{tabular}
}
\caption{\footnotesize A summary of results for different variants of QAOA (for Max-Cut) based on various combinations of initializations (equal superposition, BM-MC$_k$, projected GW, and perturbed single-cut initializations) and mixing Hamiltonian (standard, custom mixers (Section \ref{sec:customMixers}), and the mixer proposed by Egger et al. \cite{egger2020warm} for perturbed single-cut initializations). For various combinations, we state what is known regarding the convergence and worst-case approximation ratio (for depths $p\geq 0$) of the corresponding QAOA variant for graphs with non-negative edge weights.}
\label{tab:summaryTable}
\end{table*}

Lastly, we compare the two warm-start techniques above to perturbed single-cut initializations, another warm-start technique that has been considered in previous literature \cite{egger2020warm, CFGRT22}. We show that, theoretically, such initializations give better depth-0 guarantees for QAOA (compared to the two warm-start techniques previously discussed); however, we later show (Section \ref{sec:experiments}) that such initializations yield a comparatively much slower rate of convergence with increased circuit depth.

\subsection{SDP-based Relaxations}
\label{sec:GWRelaxations}
Recall that a solution to the Goemans-Williamson (GW) SDP relaxation consists of $n$ unit vectors $\{u_i \in \R^n: i \in V\}$, and their rounding algorithm uses a random hyperplane to obtain an approximation for the Max-Cut on the given graph $G$. We propose to create a warm-start to QAOA by instead rounding each of these vectors to $\R^k$ (with $k \in \{2, 3\}$) and then mapping the rounded vectors to the Bloch sphere.\footnote{An iterative rounding of the SDP solution was independently shown by Parekh and Thompson, a couple of months before our update on arXiv, for a more general setting of the Quantum Max-Cut \cite{parekh2022optimal}.}

Specifically, given $u \in \R^n$ for some positive integer $n$, and a linear subspace $A$ of $\R^n$, let $\Pi_A(u)$ denote the (Euclidean) projection of $u$ on $A$. Given $\Pi_A(u) \neq 0$, define
\[
    \Lambda_A(u) = \frac{\Pi_A(u)}{\|\Pi_A(u)\|_2},
\]
as the \emph{unit-scale} projection of $u$ on $A$. This corresponds to normalizing $\Pi_A(u)$ so it is a unit vector.

To motivate warm-starts using solutions for the Goemans-Williamson SDP, consider the following two-step process for rounding an optimal SDP solution $\{u_i: i \in [n]\}$ to a cut:
\begin{itemize}
    \item Choose a uniformly random linear subspace $A$ of $\R^n$ of dimension $k \in \{2, 3\}$, and consider the unit-scale projections $\Lambda_A(u_i) \in \R^k, i \in [n]$.
    \item Use Goemans-Williams hyperplane rounding on vectors $\Lambda_A(u_i), i \in [n]$ to get a (random) cut $M'$
\end{itemize}

That is, round the vectors $u_i \in \R^n$, $i \in [n]$ to a random $k$-dimensional subspace first, and then subsequently use the Goemans-Williamson hyperplane rounding on these $k$-dimensional vectors.

The following theorem shows that this two-step rounding is equivalent to the Goemans-Williamson hyperplane rounding on vectors $u_i$, $i \in [n]$; we include the proof in Appendix \ref{sec:proofs}.

\begin{theorem}\label{thm: subspace-hyperplane-rounding}
    Suppose we are given unit vectors $u_1, \ldots, u_n \in \R^n$ that form an optimal solution to the SDP relaxation for Max-Cut on some graph $G= (V,E)$ with $n$ vertices and non-negative weights on the edges. Suppose the GW random hyperplane rounding on $u_1, \ldots, u_n$ obtains a (random) cut $M$ of value $X$, and the two-step rounding described above produces a (random) cut $M'$ of value $Y$. Then,
    \begin{enumerate}
        \item The random variables $X = Y$. In particular, $\mathbb{E}(X) = \mathbb{E}(Y)$ and therefore, in expectation, $M^\prime$ provides a $0.878$-approximation to Max-Cut on $G$.
        \item Furthermore, the two-step rounding procedures produces a cut of value $(0.878 - \epsilon)$ times the Max-Cut value \emph{with high probability} if performed independently $\frac{\log n}{\epsilon}$ times for any constant $\epsilon \in (0, 1/2)$.
        Further, if $A_1, \ldots, A_{\frac{\log n}{\epsilon}}$ are the intermediate $k$-dimensional subspaces in these $\frac{\log n}{\epsilon}$ runs, there is at least some $A_i$ (with high probability) such that performing the hyperplane rounding on $A_i$ produces a (random) cut of average value at least $(0.878 - \epsilon)$ times the Max-Cut.
    \end{enumerate}
\end{theorem}

The last part of Theorem \ref{thm: subspace-hyperplane-rounding} illustrates that we can obtain a high-quality $k$-\dimensional projected GW solution in regards to hyperplane rounding; however, it is natural to ask if any kind of guarantee can be preserved when mapping the solution to a quantum state. By adapting a theorem by Tate et al. \cite{TFHMG20}, we can answer the question in the affirmative as seen in Corollary \ref{thm:preserveHyperplaneRatioQuantum} below.

\begin{corollary}
\label{thm:preserveHyperplaneRatioQuantum}
Let $G$ be a graph with non-negative edge weights and let $x$ be a corresponding $\kappa$-approximate projected GW solution in $\mathbb{R}^3$ with respect to hyperplane rounding.\footnote{The theorem also holds more generally for any feasible BM-MC$_2$ or BM-MC$_3$ solution.} Let $R_U(x)$ denote random uniform rotation applied to $x$, i.e., a global rotation where a uniformly selected point on the sphere gets mapped to $(0, 0, 1)$. Then initialization of QAOA with $R_U(x)$ has an (worst-case) approximation ratio of $\frac{2}{3}\kappa$ at $p = 0$, i.e., only using quantum sampling with initial state creation and no algorithmic depth for QAOA. Similarly, if $x$ is a $\kappa$-approximate projected GW solution in $\mathbb{R}^2$, initialization of QAOA with $R_U(x)$ is a $\frac{3}{4}\kappa$-approximate solution at $p=0$.

If $x$ is chosen such that it is $\kappa$-approximate with $\kappa = 0.878 - \varepsilon$ (such an $x$ is easily found via Theorem \ref{thm: subspace-hyperplane-rounding}), then, for small $\varepsilon$, this yields (worst-case) approximation ratios (for depth-0 QAOA-warmest) of $\frac{3}{4}(0.878 - \varepsilon) \approx 0.658$ and $\frac{2}{3}(0.878 - \varepsilon) \approx 0.585$  for 2-dimensional and 3-dimensional projections respectively.
\end{corollary}

A proof of Corollary\footnote{Note that Theorem 2 in \cite{TFHMG20} works with the assumption that $x$ is $\kappa$-close whereas Corollary \ref{thm:preserveHyperplaneRatioQuantum} assumes that $x$ is $\kappa$-approximate (see Section \ref{subsec:classicalAlgoPrelims}).} \ref{thm:preserveHyperplaneRatioQuantum} can be found in Appendix \ref{sec:proofs}. This motivates us to use the projected vectors $\Lambda_A(u_i), i \in [n]$ to warm-start QAOA. We then use the same scheme as proposed by Tate et al. \cite{TFHMG20} to rotate and map the solution $x^* = \{\Lambda_A(u_i)\}_{i \in [n]}$ to the Bloch sphere. More specifically, we first perform either (1) a random vertex-at-top rotation of $x^*$ (a global rotation to $x^*$ so that a random vertex $v$ coincides with (0,0,1)) or (2) a uniformly random rotation $R_U$, and apply the natural mapping from the rotated solution to the Bloch sphere to obtain a separable, unentangled state $\ket{s_0}$ \cite{TFHMG20}. Figure \ref{fig: gw-warm-starts-schematic} illustrates the two-step rounding procedure and this warm-start.

\begin{figure*}[htbp]
    \centering
    \resizebox{0.9\linewidth}{!}{%
    \begin{tikzpicture}
        \tikzset{vertex/.style = {shape=circle,draw,minimum size=30pt}}
        \tikzset{edge/.style = {->,> = latex'}}
        
        \node[vertex] (a) at  (-2,5) {$\R^n$};
        \node[vertex] (b) at  (3,5) {$\R^k$};
        \node[vertex] (c) at  (6,5) {$\R^1$};
        \node[vertex] (d) at  (3,2) {};
        \node[vertex] (e) at  (6,2) {Cut};
        
        \node (f) at (0.5, 5) {{\Huge$\ldots$}};
        \node (g) at (2.2, 3.1) {\small Rotation};
        \node at (3.3, 3.05) {\scriptsize \color{blue} III};
        \node (h) at (4.5, 1.7) {QAOA};
        \node at (4.5, 2.2) {\scriptsize \color{blue} III};
        \node (i) at (3, 2.2) {\scriptsize Initial};
        \node (j) at (3, 1.8) {\scriptsize state};
        \node at (2,6.6) {{\scriptsize Goemans-Williamson hyperplane rounding}};
        \node at (2,6.1) {{\scriptsize I}};
        \node at (0.5,3.9) {{\scriptsize $k$-dimensional rounding}};
        \node at (0.5,4.3) {{\scriptsize II, {\color{blue} III}}};
        \node at (4.5,4.1) {{\scriptsize \begin{tabular}{c}
            Random hyperplane \\
            rounding
        \end{tabular}}};
        \node at (4.5,4.6) {{\scriptsize II}};
        \node at (-3.2,5) {\scriptsize \begin{tabular}{c}
            SDP  \\
            solution
        \end{tabular}};
        \node at (7.2,5) {\scriptsize \begin{tabular}{c}
            Random  \\
            Cut 
        \end{tabular}};
        
        \draw[edge, blue] (a)  to[bend right] (b);
        \draw[edge] (b)  to[bend right] (c);
        \draw[edge] (a)  to[bend left] (c);
        \draw[edge, blue] (b)  to (d);
        \draw[edge, blue] (d)  to (e);
    \end{tikzpicture}
    }
    \caption{\footnotesize A schematic for GW warm-starts. There are three procedures to obtain a cut from an SDP solution. The first is to use Goemans-Williamson hyperplane rounding (on the top labelled I). The second (labelled II) is to do a two-step rounding through an intermediate state in $\R^k, k \in \{2, 3\}$. We prove that this two-step rounding procedure is equivalent to Goemans-Willimson hyperplane rounding in Theorem \ref{thm: subspace-hyperplane-rounding}. The third procedure is our proposed warm-start of QAOA using the SDP solution (highlighted in blue, labelled III). This procedure involves rounding the SDP solution to $\R^k$ first, then rotating this solution using uniform or vertex-at-top rotations and mapping to the Bloch sphere to get an initial state for QAOA, and finally running QAOA on this initial state.}
    \label{fig: gw-warm-starts-schematic}
\end{figure*}
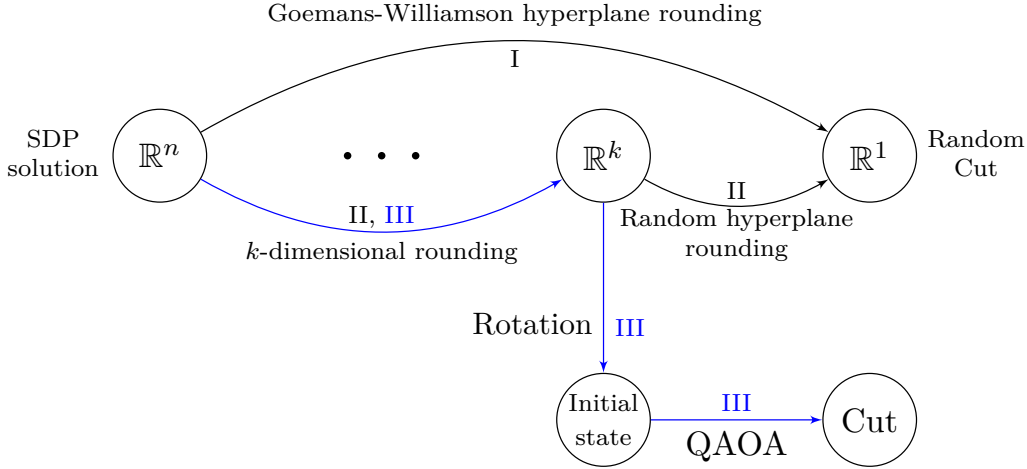

\subsection{Burer-Monteiro's Relaxation} For faster computational performance, we find a warm-start by using a locally optimal $k$-\dimensional solution $x^*$ (with $k \in \{2,3\}$) obtained by the Burer-Monteiro relaxation for Max-Cut on a $(k-1)$-sphere. Similar to projected GW solutions, we use the same rotation and quantum-mapping scheme to obtain an unentangled initial quantum state.

Theorem 2 of Tate et al. \cite{TFHMG20} shows that the same approximation ratios of $\frac{3}{4}\kappa$ and $\frac{2}{3}\kappa$ (for $k = 2$ and $k = 3$ solutions respectively) can be achieved for depth-0 QAOA-warmest if the solution is instead $\kappa$-close, i.e., $\text{BM-MC}_k(x^*) \geq \kappa \text{Max-Cut}(G)$. Using the same analysis as Goemans and Williamson \cite{GW95}, it is straightforward to show that a $\kappa$-close solution must also be $0.878\kappa$-approximate (i.e. $\text{HP}(x) \geq 0.878\kappa \text{Max-Cut}(G))$. For locally optimal\footnote{When $x^*$ is a globally optimal BM-MC$_k$ solution, it immediately follows that $x^*$ is $1$-close and $0.878$-approximate and hence {(worst-case)} approximation ratios of $\frac{3}{4}$ and $\frac{2}{3}$ (for $k = 2$ and $k = 3$ solutions respectively) are achieved for depth-0 QAOA-warmest; however, since the Burer-Monteiro relaxation is non-convex, finding such globally optimal solutions becomes intractable, especially as the number of nodes increases.} BM-MC$_k$ solutions, Mei et al. \cite{MMMO17} shows that such solutions are $\kappa$-close (and hence $0.878\kappa$-approximate) with $\kappa = 1-\frac{1}{k-1}$; thus, hyperplane rounding of such solutions yields (worst-case) approximation ratios of $0.878(0) = 0$ and $0.878\left(\frac{1}{2}\right) = 0.439$ while depth-0 QAOA-warmest can obtain {(worst-case)} approximation ratios of $\frac{3}{4}\left(0\right) = 0$ and $\frac{2}{3}\left(\frac{1}{2}\right) = \frac{1}{3}$ for $k = 2$ and $k = 3$ solutions respectively. 

Experiments in Appendix \ref{sec:GW_BMMC_Scaling} compare the (instance-specific) approximation ratios achieved by hyperplane rounding of both types of warm-start initializations discussed; in particular, projected GW SDP solutions do well (compared to $n$-\dimensional GW SDP solutions) but approximate BM-MC$_k$ solutions degrade in performance as the number of nodes increases.

\subsection{Perturbed Single Cut Initializations}
\label{sec:singleCutInitialization}
For the purposes of comparison in our numerical simulations (Section \ref{sec:experiments}), we briefly review one more type of warm-start which we refer to as perturbed single-cut initializations that other researchers \cite{egger2020warm, CFGRT22} have used for QAOA; more details regarding this approach can be found in (Appendix \ref{sec:singleCutAppendix}). Given a regularization angle $\theta^* \in [0,\pi/2]$ and a cut $(S, V \setminus S)$, one can initialize the initial quantum state so that qubits lie along the $xz$-plane of the Bloch sphere with vertices in $S$ and $V\setminus S$ being initialized at an angle $\theta^*$ away from the north and south poles of the Bloch sphere respectively; such a regularization angle aims to circumvent issues regarding reachability\cite{egger2020warm}.\footnote{It can be shown that if the initial qubit position lies on the poles of the Bloch sphere, e.g. if the initial qubit position is $\ket{0}$ (north pole), then it is guaranteed that that qubit will be measured to be $\ket{0}$ as the end of the QAOA algorithm if custom mixers (as described in Section \ref{sec:customMixers}) are used.}


Though not stated in the related works, if the warm-starts are initialized close to single-cuts, then they retain the classical approximation factors of the respective single-cut:

\begin{obs}
    \label{thm:singleCutAR}
    Let $G = (V, E)$ be a graph, and suppose a cut $(S, V \setminus S)$ is an $\alpha$-approximation to Max-Cut$(G)$, and is used to initialize a warm-start using regularization angle of $\theta^* \in [0,\pi/2]$. Then, a quantum measurement (i.e. depth-0 QAOA) of this state yields an approximation ratio of at least
    $\alpha \cdot \cos(\theta^*/2)^{2|V|}$.
\end{obs}

It is easy to see that such warm-starts approach an approximation ratio (at $p=0$) of 0.878, when initialized randomly with the distribution of cuts obtained using the Goemans-Williamson algorithm, as the regularization angle $\theta^*$ approaches zero. We discuss a proof in the Appendix \ref{sec:proofs} to show this. Although Proposition \ref{thm:singleCutAR} suggests it may be the superior warm-start technique for small regularization angle $\theta^*$; convergence is either lacking or slow empirically depending on the mixer used (see Table \ref{tab:summaryTable}).

\section{Numerical Simulations and Experiments}
\label{sec:experiments}

In this section, we demonstrate the importance of using a suitable warm-start and show that with such warm-starts, QAOA-warmest outperforms the Goemans-Williamson Max-Cut algorithm as well as standard QAOA \cite{FGG14} and QAOA-warm \cite{TFHMG20}. In particular, we show that QAOA-warmest does at least as well as these other algorithms at depths $p \geq 4$ for nearly every instance in our instance library. We consider comparisons with respect to a recent warm-starts approach of Egger et al. \cite{egger2020warm} in Appendix \ref{sec:additionalExperiments}.

As discussed earlier (Section \ref{sec:warmstarts}), for positive-weighted graphs, perturbed single-cut initializations have better depth-0 guarantees compared to the projected GW warm-starts that we propose. While possibly more advantageous at extremely low depths ($p=0,1$), we show that QAOA with perturbed single-cut initializations empirically converge to an optimal cut very slowly with increased circuit depth; for small enough regularization angle $\theta^*$, the convergence towards an optimal cut is not even perceivable at the circuit depths tested. Meanwhile, with suitable warm-starts, the convergence is much quicker: for over $98\%$ of instances tested, we found that depth-8 QAOA-warmest with BM-MC$_2$ initializations yields nearly optimal cuts.

Lastly, we consider the effects of noise on QAOA and its variants in the context of actual quantum devices (i.e. the IBM-Q Guadalupe and Quantinuum H1-1 devices). We show that for QAOA and its variants, the noise from these devices flattens the landscape without significantly altering the location of local extrema. Additionally, we show that even with noise, QAOA-warmest (with suitable warm-start) maintains a significant fraction of its expected solution quality, which suggests it may be useful for near-term NISQ devices (potentially with some noise mitigation) \cite{P18}. 

\subsection{Simulation Details}
 For our simulations, we use the CI-QuBe library\footnote{\url{https://github.com/swati1729/CI-QuBe}} \cite{CI-QuBe2021} which contains graphs up to 11 nodes using a variety of random graph models (Erd\H{o}s-R\'enyi, Barabasi Albert, Dual of Barabasi-Albert, Watts-Strogatz, Newman-Watts-Strogatz, and random regular graphs) and edge weight distributions. These instances, which we refer to as $\mathcal{G}$, have a varied distribution of various graph properties, which is important when testing heuristics and algorithms for solving this problem.
 
In our simulations, for each instance, we first find five locally approximate solutions to BM-MC$_2$ and keep the best (in terms of the BM-MC$_2$ objective value). We do the same for BM-MC$_3$. Similarly, for each instance, we solve the GW SDP, perform 5 projections to random 2-dimensional subspaces, and keep the best (in terms of the BM-MC$_2$ objective); this process is repeated (using the same GW SDP solution) with projections to 3-dimensional subspaces.  Next, for both the best BM-MC$_2$ and best BM-MC$_3$ solution, and for both of the best projected GW solutions (in 2 and 3 dimensions), we perform 5 different vertex-at-top rotations and 5 different uniform rotations, yielding 40 different initial warm-started quantum states per instance. We run QAOA-warm and QAOA-warmest using all 40 of these warm-started states and, for each combination of \dimension and rotation scheme, record which one performed the best in terms of (instance-specific) approximation ratio (as defined in Section \ref{sec:approxRatio}). Finally, we run standard QAOA on the instance.

For each run for each variant of QAOA, we initialize the variational parameters $\gamma$ and $\beta$ close to zero\footnote{For standard QAOA, many optimizers will immediately terminate if initialized exactly at the origin due to the presence of a saddle point. Instead, each variational parameter $(\gamma_1,\dots,\gamma_p,\beta_1,\dots,\beta_p)$ is initialized by sampling uniformly from the interval $[-0.0001,0.0001]$.} and each run terminates when the difference in successive values of $F_p(\gamma,\beta)$ in the optimization loop is less than $\bar{W} \times 10^{-6}$ where $\bar{W}$ is the sum of the absolute values of the edge weights.

To simplify the results, the figures and tables in this section will only consider runs of QAOA-warmest that use BM-MC$_2$ initializations with vertex-at-top rotations. This choice is due to runtime considerations and to allow for easier comparisons with previous related literature \cite{TFHMG20, egger2020warm}; more details on the results and the choice of this decision can be found in Appendix \ref{sec:additionalExperiments}.

Additionally, for conciseness, in this section we will use ``approximation ratio" to mean the instance-specific approximation ratio as described in Section \ref{sec:approxRatio}.

\begin{table*}[htbp]
\resizebox{\linewidth}{!}{%
\begin{tabular}{|c|c?c|c|c|c?c|c|c?c|c|c?c?}
        \Xhline{2\arrayrulewidth}& 1st Best & \multicolumn{4}{c?}{QAOA-warmest} & \multicolumn{3}{c?}{Standard QAOA} & \multicolumn{3}{c?}{GW} & \multirow{2}{*}{Tie}\\\cline{2-12}
     &2nd Best &  \textbf{*} & Standard & Warm & GW & Warmest & Warm & GW & Warm & Warmest & Standard & \\\Xhline{3\arrayrulewidth}
   \multirow{4}{*}{\makecell{Positive\\Weighted \\ Graphs}}&  p=1 & \textbf{90.3 \%} & 0.69\% & 18.85\% & 15.22\% & 0.0\% & 0.0\% & 0.0\% & 0.17\% & 9.51\% & 0.0\% & 55.53\%\\\cline{2-13}
    & p=2 & \textbf{98.1\%} & 0.69\% & 20.24\% & 25.08\% & 0.0\% & 0.0\% & 0.0\% & 0.0\% & 1.9\% & 0.0\% & 52.07\%\\\cline{2-13}
    & p=4 & \textbf{100\%} & 8.65\% & 17.64\% & 20.58\% & 0.0\% & 0.0\% & 0.0\% & 0.0\% & 0.0\% & 0.0\% & 53.11\%\\\cline{2-13}
    & p=8 & \textbf{100\%} & 25.77\% & 5.01\% & 2.94\% & 0.0\% & 0.0\% & 0.0\% & 0.0\% & 0.0\% & 0.0\% & 66.26\%\\\Xhline{2\arrayrulewidth}
    \multirow{4}{*}{All Graphs} & p=1 & \textbf{90.6\%} & 0.35\% & 17.96\% & 17.43\% & 0.0\% & 0.0\% & 0.0\% & 0.53\% & 8.84\% & 0.0\% & 54.86\%\\\cline{2-13}
    & p=2 & \textbf{98.1\%} & 0.44\% & 20.17\% & 24.86\% & 0.0\% & 0.0\% & 0.0\% & 0.26\% & 1.68\% & 0.0\% & 52.56\%\\\cline{2-13}
    & p=4 & \textbf{99.6\%} &9.11\% & 19.2\% & 18.93\% & 0.0\% & 0.0\% & 0.0\% & 0.08\% & 0.26\% & 0.0\% & 52.38\%\\\cline{2-13}
    & p=8 & \textbf{99.6\%} & 27.16\% & 7.96\% & 3.18\% & 0.17\% & 0.0\% & 0.0\% & 0.26\% & 0.0\% & 0.0\% & 61.23\%\\\Xhline{2\arrayrulewidth}
\end{tabular}
}
\caption{\footnotesize For each Max-Cut algorithm (Goemans-Williamson, standard QAOA, QAOA-warmest, and QAOA-warm), we report the percentage of instances for which it did the best and second-best (in terms of approximation ratio). Both QAOA-warm and QAOA-warmest use BM-MC$_2$ warm-starts. 
There is a tie (last column) if the top two algorithms have approximation ratios that differ by no more than 0.01. QAOA-warmest is a part of every tie. Each instance is either accounted for in ``Tie" or the other columns. For the column labeled \textbf{*}, we report, for each circuit depth, the percentage of instances for which QAOA-warmest was within 0.01 approximation ratio of the best algorithm.}
\label{fig:pieChartTable}

\end{table*}

\begin{figure*}[t]
    \centering
    \includegraphics[scale=0.55]{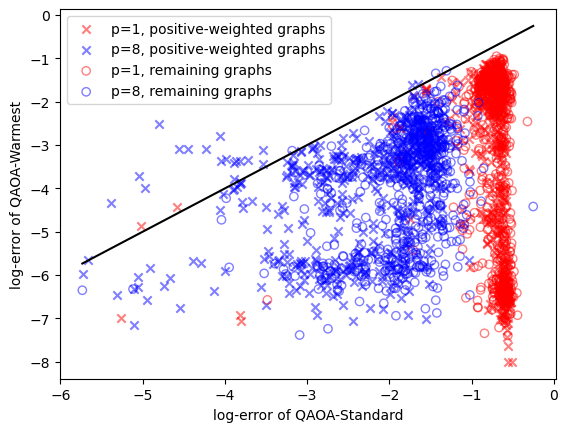}\includegraphics[scale=0.55]{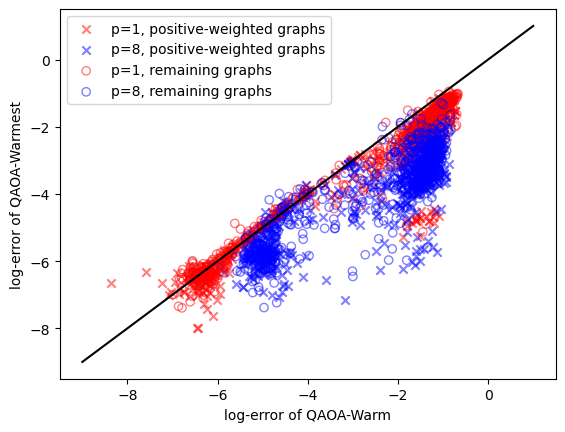}
    \caption{\footnotesize For both plots, we compare the log-error of QAOA-warmest to both QAOA-warm (right) and standard QAOA (left). BM-MC$_2$ warm-starts are used for both approaches. Each marker in the plot corresponds to a combination of instance (from our graph ensemble $\mathcal{G}$) and circuit depth (either $p=1$ or $p=8$) with the shape of the marker being used to denote if the instance has only positive edge weights or not. Points below the black line correspond to instances where QAOA-warmest performs better than the other algorithm being compared.}
    \label{fig:scatter}
\end{figure*}

\subsection{Comparing QAOA-warmest to Other Methods}
\label{sec:QAOAWarmestCompare}
In Table \ref{fig:pieChartTable}, we show the proportion of graphs where each Max-Cut algorithm (GW and variants of QAOA) performs the best for varying values of depth $p=1,2,4,8$. We observe that for nearly all instances, QAOA-warmest beats or performs as well as every other QAOA variant considered and eventually performs at least as well as GW as the circuit depth increases. We note that at $p=8$, QAOA-warmest beats GW on all but three instances but this is easily rectified with a suitable vertex-at-top rotation. Also at $p=8$, QAOA-warmest outperforms standard QAOA on all but two instances but the gap in  approximation ratio is less than 0.02. More information regarding these five instances can be found in Appendix \ref{sec:additionalExperiments}. 

We next report the improvement in  approximation ratios when considering standard QAOA, QAOA-warm, and QAOA-warmest with circuit depths $p=1,8$. For convenience, for any Max-Cut algorithm, we define the  approximation error (AE) by $\text{AE} = 1 - \text{AR}$ where AR is the  approximation ratio. Additionally, we will refer to $\log_{10}(\text{AE})$ as the log-error. Figure \ref{fig:scatter} gives a comparison of log-errors achieved for various instances. All points below the $x=y$ solid line indicate instances where QAOA-warmest beats either standard QAOA or QAOA-warm. Note that due to the plots being log-scaled, being below -2 on each axis corresponds to having an  approximation ratio of at least 0.99. For both plots, we see that higher  approximation ratios can be achieved for positive-weighted graphs (cross-marks) and that QAOA-warmest performs significantly better for most instances. When comparing QAOA-warmest and standard QAOA at various circuit depths (red v/s blue), we see that the performance for both standard QAOA and QAOA-warmest improves at $p=8$; however, this phenomenon is not that apparent for QAOA-warm (which is known to plateau in performance with increased circuit depth for small instances).

\begin{figure}[htbp]
    \centering
    \includegraphics[width=\columnwidth]{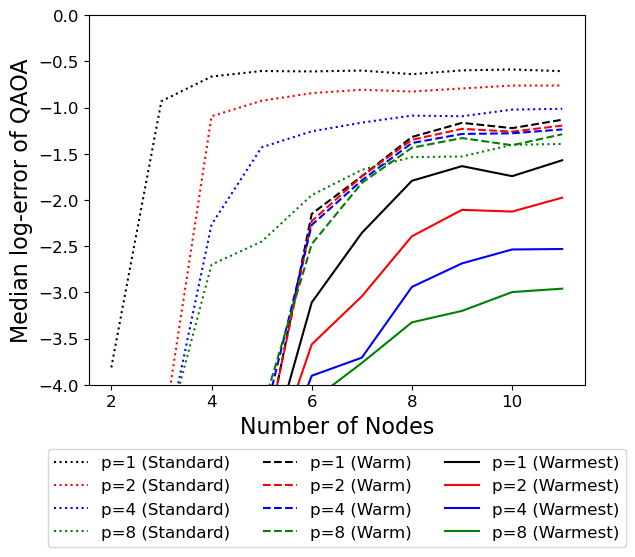}
    \caption{\footnotesize Trends in median log-error of standard QAOA (dotted), QAOA-warm (dashed), and QAOA-warmest (solid) as one varies the number of nodes and circuit depths; the median is taken across graphs in instance library $\mathcal{G}$. BM-MC$_2$ warm-starts are used for both QAOA-warm and QAOA-warmest.}
    \label{fig:trends}
\end{figure}

Next, we show the trend in  approximation quality with increase in the number of nodes $n$ and the depth of the circuit $p$, in Figure \ref{fig:trends}. We see that, across all node sizes, that circuit depth plays an important role in how good an  approximation ratio one can expect to achieve using QAOA-warmest. It is clear that QAOA-warmest has superior (median) performance compared to the other algorithms for every combination of circuit depth and node-size. We remark that in contrast, an increased circuit depth resulted in only a marginal improvement in the  approximation ratio for QAOA-warm, bolstering our claim that custom mixers are crucial to the improvement in performance of QAOA.

In Figure \ref{fig:convergenceComparison}, we compare the convergence rates of standard QAOA and QAOA-warmest with various initilializations: BM-MC$_k$, {perturbed} single-cut initializations (as described in Section \ref{sec:singleCutInitialization}), and uniform random initializations\footnote{Here, a uniform random initialization refers to a separable state that is randomly created by (independently) picking a position on the surface of the Bloch sphere uniformly at random for each qubit, and then tensorizing the qubits. Although the phase of each qubit does not effect the expected cut value obtained from QAOA-warmest (as discussed in Section \ref{sec:customMixers}), it should be noted that the distribution of initializations obtained from sampling from the surface of the Bloch sphere and then removing the phases is different than the distribution of initializations obtained from sampling uniformly from the portion of the great-circle formed by the intersection of the Bloch sphere surface with the $xz$-plane with positive $x$-coordinate.  Nonetheless, we have verified (via numerical simulation) that QAOA-warmest performs similarly between the two randomization schemes.} with custom mixers for a random 10-node instance in our graph library. Consistent with what is seen in the other figures, we see that standard QAOA, QAOA-warmest, and uniform random initializations  quickly achieve high  approximation ratios at relatively low circuit depths, with the BM-MC$_2$ initialization doing the best amongst the three across all circuit depths tested. On the other hand, perturbed single-cut initializations do not converge as quickly; in particular, when $\theta^*$ is small ($\theta^* \in \{0.1, 0.01\}$) hardly any improvement in the  approximation ratio is observed at all. For larger regularization angles ($\theta^* = 0.5$), we do see worse performance at low depths ($p=0,1$) as well as a noticeable increase in performance with increased circuit depth; however, the amount of this increase is small compared to achieved by QAOA-warmest which begins to outperform the perturbed single-cut initialization (with $\theta^*=0.5$) for $p \geq 2$. We find similar qualitative results for most other instances in our graph ensemble $\mathcal{G}$.

Table \ref{tab:convergenceAgg} provides a more aggregated view of the convergence of QAOA-warmest with different choices of initializations across the entire instance library $\mathcal{G}$. For each combination of initialization method and circuit depth, the table states the percentage of instances in the library which achieved an instance-specific approximation ratio of 99.0\% of higher. The data for the perturbed single-cut initializations were obtained as follows: for each instance, we obtained an optimal solution to the GW SDP relaxation, we performed 100 hyperplane roundings on the optimal SDP solution to obtain 100 cuts, we discarded all cuts whose value is more than $0.98 \cdot \text{\sc Max-Cut}(G)$, and we used the best remaining cut to create a perturbed single-cut initialization with regularization angle $\theta^* = 0.1$ (as described in Section \ref{sec:singleCutInitialization}); the discarding of cuts with very high values were done in order to ensure that, in the case of a high instance-specific approximation ratio, such a ratio can be partly attributed to the quantum circuit and not just the initial cut itself. When using the BM-MC$_2$ initializations (with vertex-at-top rotations), there are steady improvements with increased circuit depths with 42.3\% of the instances achieving an instance-specific AR of 99.0\% at depth-0; this percentage increases to 98.1\% at depth-8. With the standard QAOA initialization, $\ket{+}^\otimes$, none of the instances achieve an instance-specific AR of 99.0\% or more; it is not until depth $p=8$ that we see a considerable fraction of the graphs (39.4\% respectively) achieving such an AR. As for the perturbed single-cut initialization with small regularization angle, we find that nearly none of the instances achieve an instance-specific AR of 99.0\% with the exception of a few instances at depth $p=8$.

\begin{table}[!t]
\centering
\resizebox{\linewidth}{!}{%
\begin{tabular}{|c|l|l|l|l|l|}
\hline
\multicolumn{1}{|l|}{}   & $p=0$    & $p=1$    & $p=2$    & $p=4$    & $p=8$    \\ \hline
BM-MC$_2$  & 42.3\% & 57.8\% & 75.0\% & 91.9\% & 98.1\%   \\ \hline
\begin{tabular}[c]{@{}c@{}}Standard\\ Initialization\\ $\ket{+}^{\otimes n}$\end{tabular}   & 0\%    & 0.6\%  & 2.4\%  & 8.7\% & 39.4\% \\ \hline
\begin{tabular}[c]{@{}c@{}}Single Cut\\ Initialization\\ with $\theta^* = 0.1$\end{tabular} & 0\%    & 0\%    & 0\%    & 0\%    & 0.7\%  \\ \hline
\end{tabular}
}
\caption{\footnotesize \label{tab:convergenceAgg} The percentage of instances for which QAOA-warmest achieves an instance-specific AR of 99.0\% for each combination of circuit depth and  initialization method  (standard initialization, BM-MC$_2$ initialization, perturbed single-cut initialization with $\theta^* = 0.1$).}
\end{table}

\begin{figure}[!t]
    \centering
    \includegraphics[scale=.55]{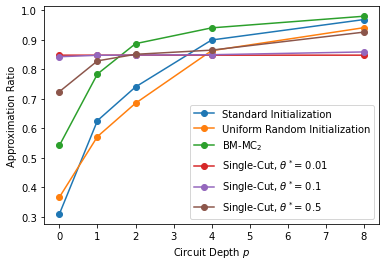}
    \caption{\footnotesize Instance specific approximation ratios achieved by QAOA-warmest with various types of initializations: standard initialization (equivalent to standard QAOA), BM-MC$_2$ initialization, perturbed single-cut initialization, and uniform random initializations) for a randomly selected 10-node instance. For QAOA-warmest, we used a BM-MC$_2$ initialization with a vertex-at-top rotation; here we intentionally chose the worst vertex (i.e. the one with worst AR at depth-0) to better illustrate the convergence rate. For the perturbed single-cut initialization, we chose a cut that obtains an instance-specific approximation ratio of 0.848 and created initial quantum states using regularization angles of $\theta^* = 0.01, 0.1,$ and $ 0.5$ radians. For the uniform random initializations, five such initializations were created and only the best one was kept (i.e. the one with best AR at depth-0).} 
    \label{fig:convergenceComparison}
\end{figure}

\begin{figure}[!t]
    \centering
    \includegraphics[trim=80 30 100 100, clip,scale=0.2]{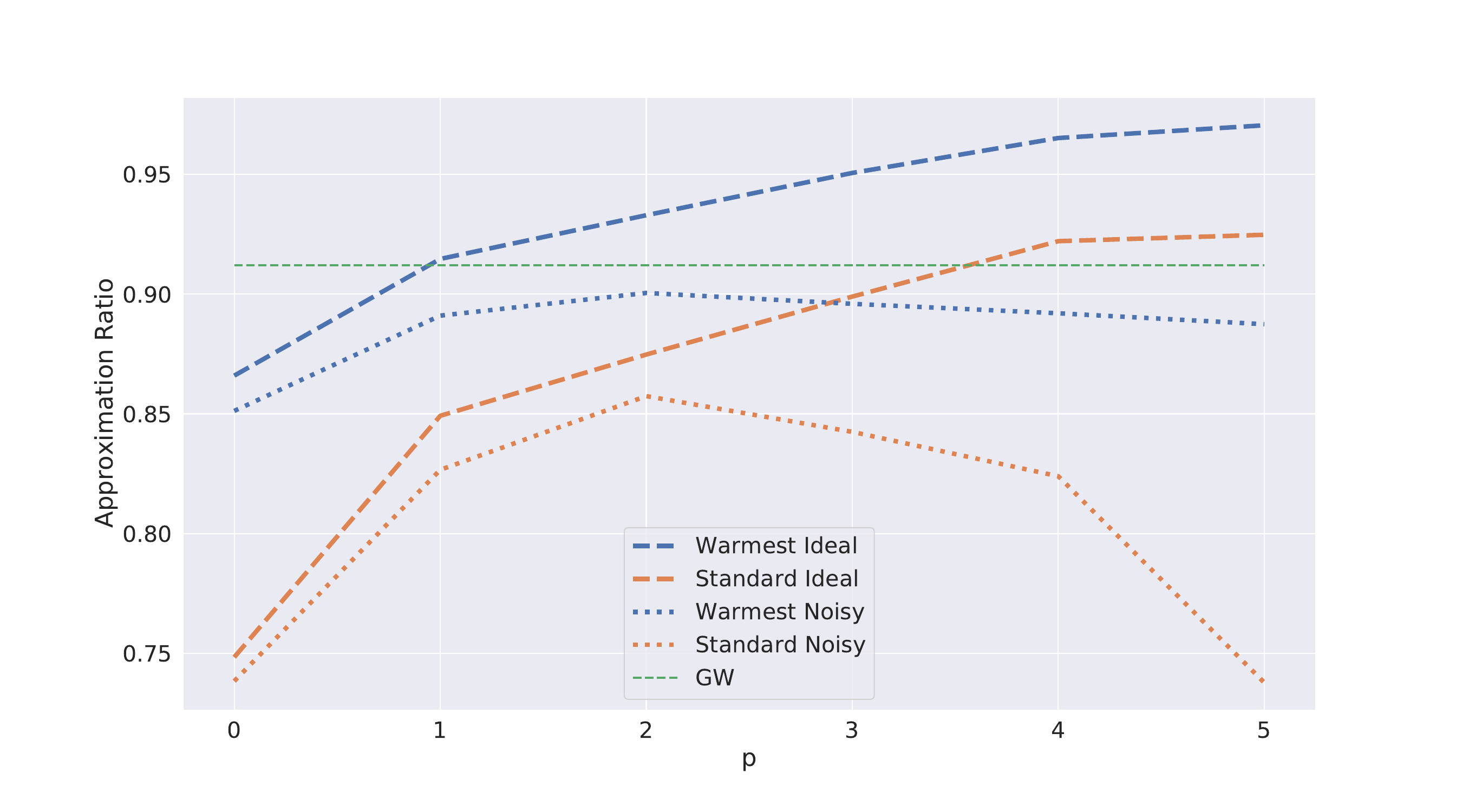}
    \caption{\footnotesize Performance of QAOA-warmest (with BM-MC$_2$ warm-starts) and standard QAOA as a function of QAOA depth for an ideal (dashed) and noisy simulation (dotted). For the chosen 20-node graph, GW acheives an  approximation ratio of 0.912, while in the ideal case, QAOA-warmest outperforms GW for $p\geq 2$ while standard QAOA requires $p>4$. The noise simulation is based on calibration data from IBM-Q's Guadalupe device.}
    \label{fig:Karlov_20}
\end{figure}

\begin{figure}[!t]
    \centering
    \includegraphics[scale=0.6]{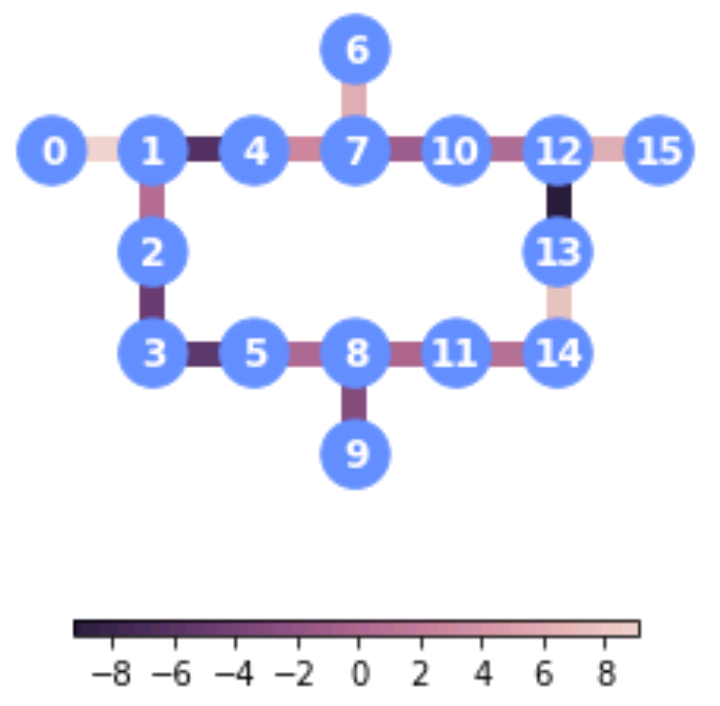}
    \caption{\footnotesize IBMQ Guadalupe device which shows the physical connectivity of qubits. We choose a native graph which matches this connectivity and random weights as indicated by color.}
    \label{fig:Guadalupe_device}
\end{figure}

\subsection{QAOA-warm With Noise}
In addition to the theoretical (noise-less) behavior of QAOA-warmest, we also demonstrate its performance with several example cases using noise models and experiments on IBM-Q hardware. For both the ideal and noisy simulation, we use IBM's Qiskit software package~\cite{qiskit}. In the case of the noisy simulation, we exercise the capability of Qiskit to pull calibration data directly from the Guadalupe device and use it to construct a noise model for use in the simulator. In principle, this combination of actual hardware calibration and noise simulation should predict the behavior of the device. However, the noise models themselves have inherent assumptions that the noise itself is uncorrelated and only directly models effects such as single and two-qubit gate errors, finite qubit lifetime and dephasing time, and readout noise. While these serve as a good starting point to model the noise in a quantum device, as shown in the Figure  \ref{fig:Guadalupe_comparison}, there is significant disagreement between the noise simulation and the actual hardware results. This disagreement is mainly attributed to the assumptions mentioned earlier, specifically the assumption of uncorrelated noise, where physical hardware experience significant crosstalk. For an example demonstration of how typical noise models are utilized, see Appendix~\ref{sec:app_noise}. 

We show in Figure~\ref{fig:Karlov_20} the performance of QAOA-warmest and standard QAOA on an instance generated via a construction by Karloff \cite{K99}; this unweighted graph is chosen due to the fact that it is a small graph that achieves a GW approximation ratio of 0.912 (see Appendix \ref{sec:twentyNodeGraph}) which is close to the lower bound of 0.878 provided by the GW algorithm. In contrast, both QAOA-warmest and standard QAOA are able to outperform this approximation ratio, under ideal, noiseless conditions. However, note that QAOA-warmest outperforms standard QAOA for all QAOA depths and outperforms GW after $p>1$. We also consider a noise model utilizing Qiskit's built in modules \cite{qiskit} and use calibration data in order to simulate IBM's Guadalupe device. We note that QAOA-warmest outperforms standard QAOA for all noisy simulations, using the same fixed noise model. 

In addition to this device-focused noise simulation, we also run QAOA-warmest on a native hardware graph matching IBM's Guadalupe device. In general, the connectivity of the graph and its matching to physical qubit hardware connectivity plays a key role in performance due to the overhead of inserting swap operations in order to compensate for limited connectivity \cite{harrigan2021quantum}. Therefore, the simplest graph is a so-called native graph, which is a graph with exactly the same connectivity as the underlying physical qubit device. This graph is shown in Figure \ref{fig:Guadalupe_device}. We assign randomly chosen weights to each edge chosen from a uniform distribution $[-10,10]$. Finding the Max-Cut solution to this graph can still be done by brute force and for a fixed choice of randomly weighted edges, we find the Max-Cut value to be approximately 33.96209.

We show the results of QAOA-warmest and standard QAOA in an ideal simulation and on hardware in {Figure} \ref{fig:Guadalupe_results}. The color scale is shared across all plots,  showing that QAOA-warmest is able to find larger cut values as compared to standard QAOA, both in simulation and on actual hardware. For hardware results, we apply the efficient SPAM noise mitigation strategy based on a CTMP strategy~\cite{Bravyi21,barron2020measurement}.

\begin{figure*}[!t]
    \centering
    \rotatebox{90}{\hspace{1.1cm} Warmest \hspace{1.5cm} Standard}\hspace{0.1cm}\rotatebox{90}{\hspace{1.8cm} $\beta$ \hspace{2.8cm} $\beta$}\hspace{0cm}
    \vspace{-0.1cm}
    \includegraphics[scale=0.30]{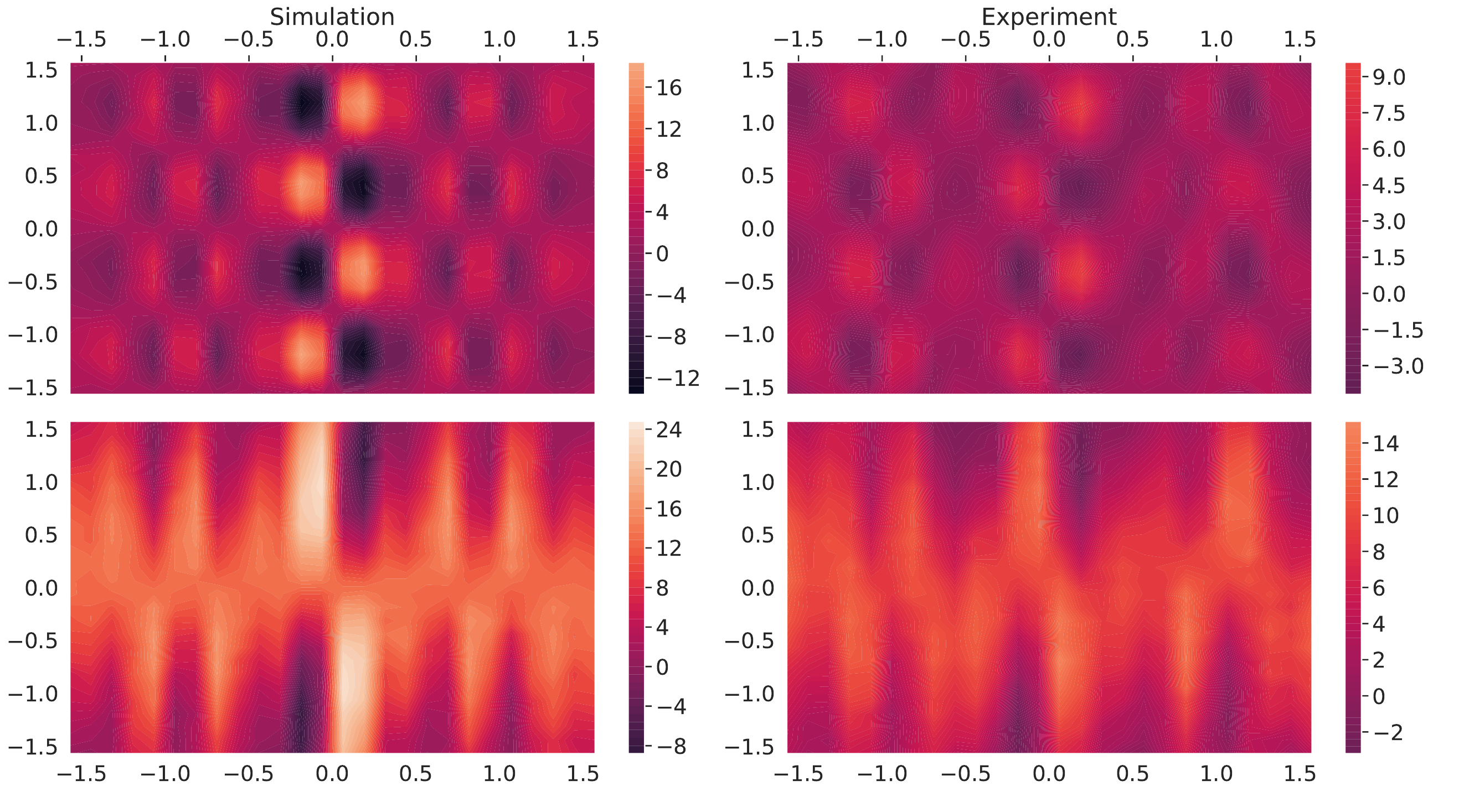}
    \vspace{-0.5cm}
    $$\gamma \hspace{7cm} \gamma$$
    \caption{\footnotesize Performance of QAOA-warmest (with BM-MC$_2$ warm-starts) compared to standard QAOA in an ideal simulation and on IBM-Guadalupe hardware. Each subfigure is a scan of $p=1$ parameters $\beta$ vs $\gamma$, brighter regions indicating values which result in a larger cut. All figures share the same absolute color scale.}
    \label{fig:Guadalupe_results}
\end{figure*}

\begin{figure*}[!t]
    \centering
    \includegraphics[width=\linewidth]{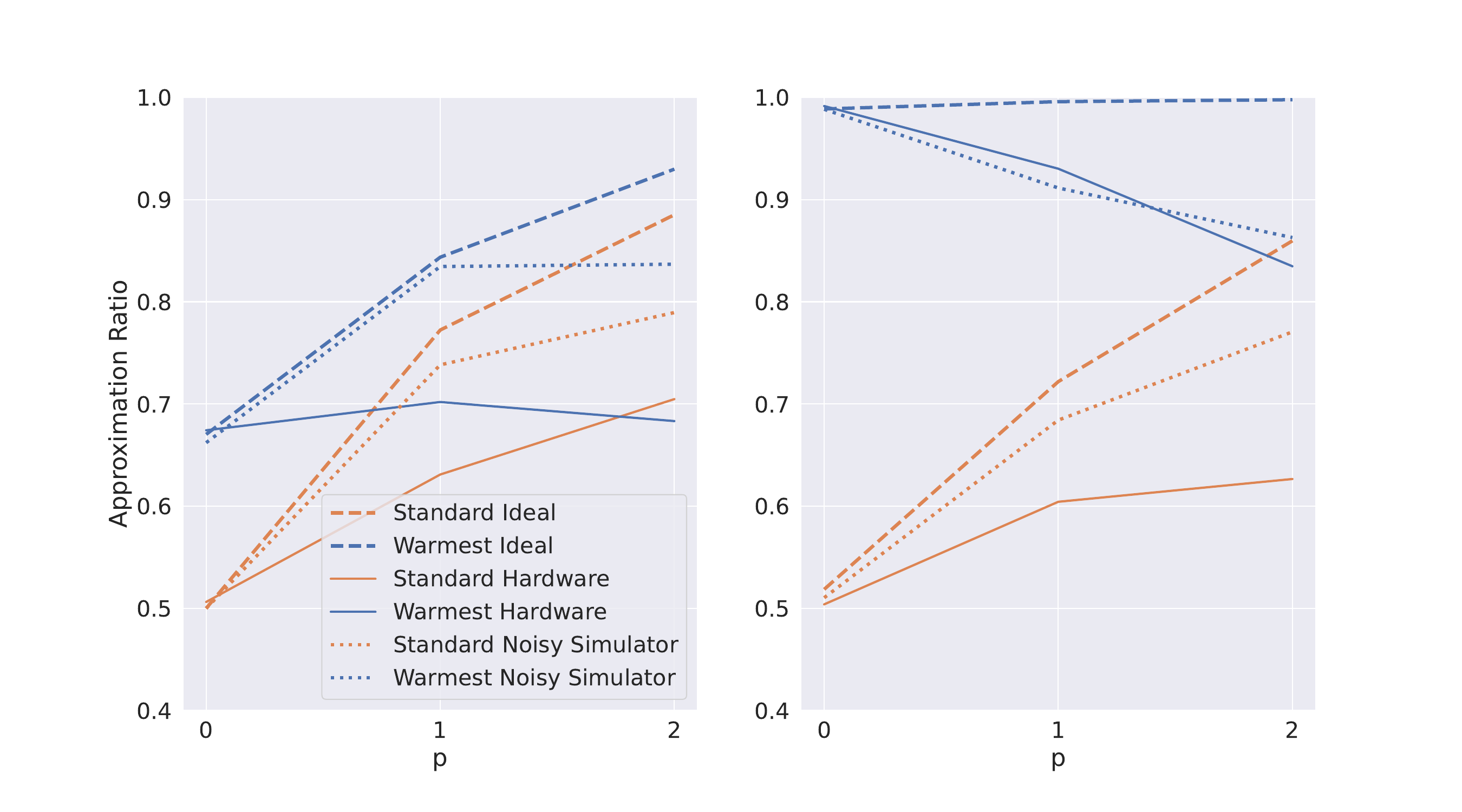}
    \caption{\footnotesize Performance of QAOA-warmest (with BM-MC$_2$ warm-starts) compared to standard QAOA in an ideal simulation (dashed), noisy simulation (dotted) and on IBM Guadalupe hardware (solid). Each subplot considers a different native hardware graph with randomly selected weights as well as a different choice of initialization procedures. (Left) For a random initialization of the classically informed QAOA-warmest start rotation. (Right) For an efficiently selected, optimal choice for the classically informed QAOA-warmest start rotation.}
    \label{fig:Guadalupe_comparison}
\end{figure*}

In order to demonstrate the scaling of QAOA-warmest, we also show results for depths $p=0,1,2$ in ideal simulation, noisy simulation, and on hardware, as shown in Figure \ref{fig:Guadalupe_comparison}. We define $p=0$ to simply mean the preparation and measurement of the initial state.\footnote{The warm-starts come from $k = 2$ or $k = 3$ solutions whereas as the GW algorithm uses $n$-\dimensional solutions. Moreover, the way cuts are determined are different (hyperplane rounding vs quantum measurement) so we should expect there to be a difference in approximation ratios.} In the case of QAOA-warmest, this directly demonstrates the ability of the QPU to create and measure the classically suggested cut.

In addition, Figure \ref{fig:Guadalupe_comparison} shows results for two different choices of the state initialization for QAOA-warmest. The left plot shows the result of applying a uniform rotation in the classical preprocessing stage whereas the right shows the result of using the best vertex-at-top rotation amongst the 16 possible vertices, i.e., the rotation that gives the largest approximation ratio at $p=0$. These two plots clearly show the importance of initializing the initial quantum state in an optimal way. Another important point shown in these plots is that small scale QAOA problems on 16 nodes, are nearly exactly solved when a suitable vertex-at-top rotation is chosen. When the best vertex-at-top rotation is used, the use of QAOA actually shows a decrease in solution quality on hardware. This is due to the inherent noise on the device and the fact that the solution quality is nearly optimal in the initial state. The presence of noisy two-qubit gates in further layers of the algorithm (32 CNOT gates per layer), overwhelm the small benefit of the algorithm itself for these small problems. A remaining goal then is to find native graphs on hardware for large systems, while also offering sufficiently low error rates, in order to demonstrate improved solutions with an optimally chosen initial quantum state and increased algorithmic depth ($p>0$).

Finally, we show results for QAOA-warmest run on Quantinuum hardware in Figure \ref{fig:quantinuum}. This 20-ion linear trap allows for arbitrary qubit connectivity and thus has no overhead associated with mapping a specific graph to the hardware. In this case, we again consider the 20-node Karloff instance graphs used in Figure \ref{fig:Karlov_20}, but here we use the GW warmest start initialization. Notably we utilize a uniform rotation which gives a large initial approximation ratio (at $p=0$) and while hardware cannot improve on this initial state, the degradation is small considering that each $p$ layer requires 90 two-qubit ZZ interactions (among many other single qubit operations). We also note the close agreement of the noisy simulator to the actual hardware results. In order to reduce the cost of these hardware runs, we only consider a single objective function evaluation (with 1000 shots), using noiseless simulations to find the optimal $\gamma_i,\beta_i$ at each $p$ depth. Even with these considerations, we see that the GW Warmest initialization outperforms the average GW performance on hardware up to $p\leq2$. These results indicate that current quantum hardware is very close to demonstrating improvement over the Goemans-Williamson algorithm for Max-Cut on known hard instance graphs when using the QAOA-warmest initialization procedure, and it already outperforms the average performance of GW on this particular graph.

\begin{figure}[!t]
    \centering
    \includegraphics[trim=80 30 100 100, clip,scale=0.2]{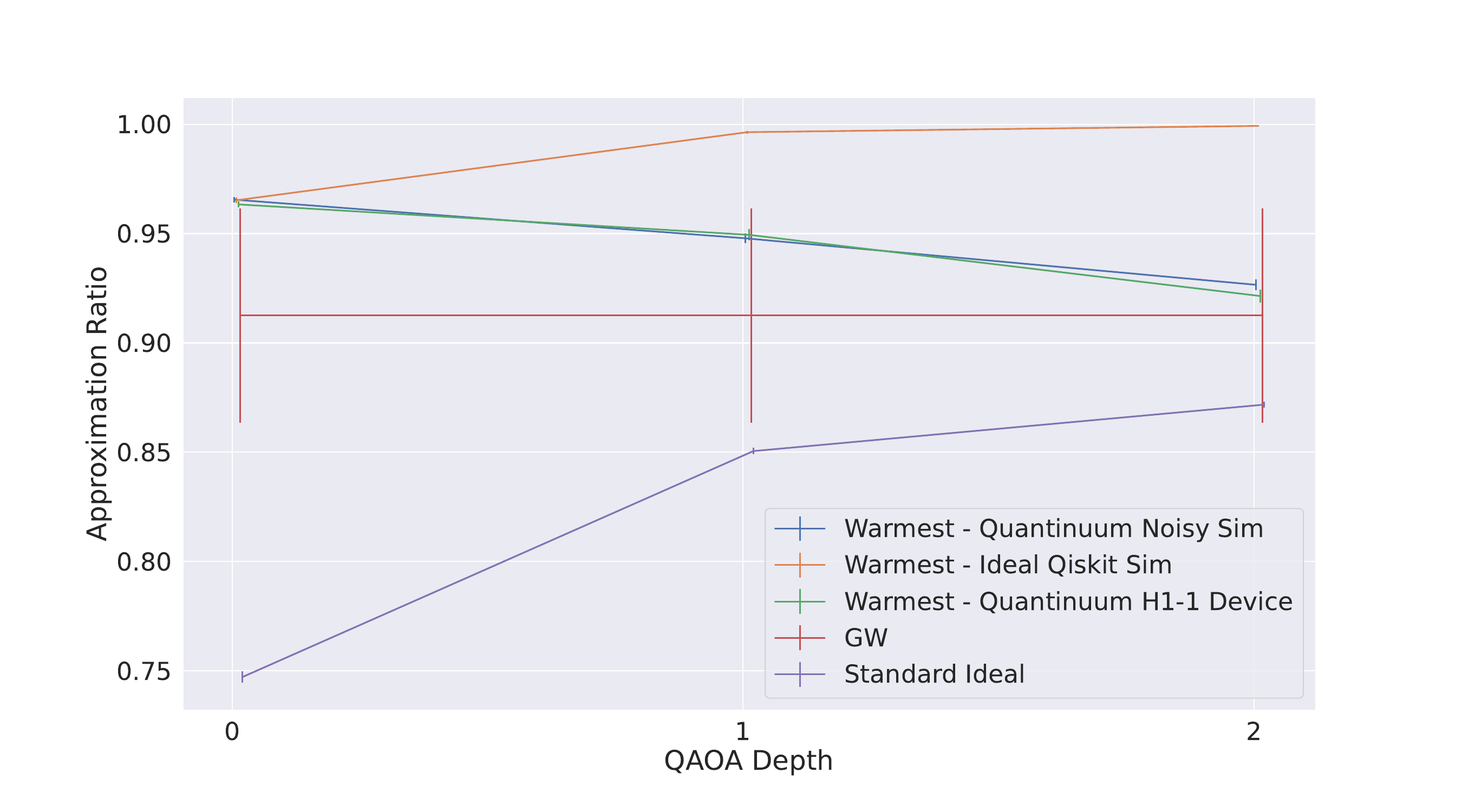}
    \caption{\footnotesize QAOA-warmest performance on Quantinuum simulators and hardware. The 20-node Karloff instance considered here is directly mapped to the fully-connected Quantinuum 20-ion hardware. In contrast to Figure \ref{fig:Karlov_20}, here we use a GW warmest start and find that this particular initialization outperforms GW on average for $p\leq2$.}
    \label{fig:quantinuum}
\end{figure}

\begin{figure}[!t]
\centering
    \includegraphics[scale=0.45]{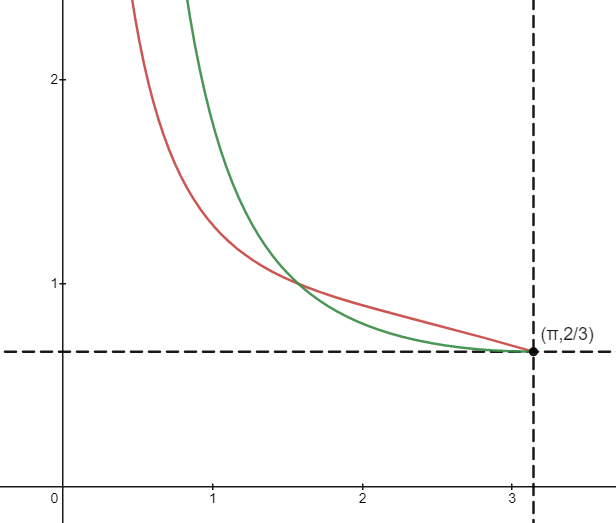}
 \caption{\footnotesize In red, the function $\frac{ \frac{1}{2}\left(1 - \frac{\cos \theta}{k}\right)}{\theta/\pi}$ that is minimized in the proof of Corollary \ref{thm:preserveHyperplaneRatioQuantum} with $k=3$. In green, a similar function $\frac{ \frac{1}{2}\left(1 - \frac{\cos \theta}{k}\right)}{\frac{1}{2}(1-\cos(\theta)}$ that is minimized in the proof Theorem 2 of \cite{TFHMG20} (where the BM-MC$_k$ objective is compared to the maximum cut instead of the expected cut value from hyperplane rounding) with $k=3$. Over the interval $[0,\pi]$, both achieve a minimum value of $2/3$ at $\theta = \pi$. The corresponding plots for $k=2$ are similar but instead both functions reach a minimum value of $3/4$ at $\theta = \pi$.}
 \label{fig:preserveHyperplaneRatioQuantum}
\end{figure}

\begin{figure*}[!t]
    \centering
    \includegraphics[scale=0.25]{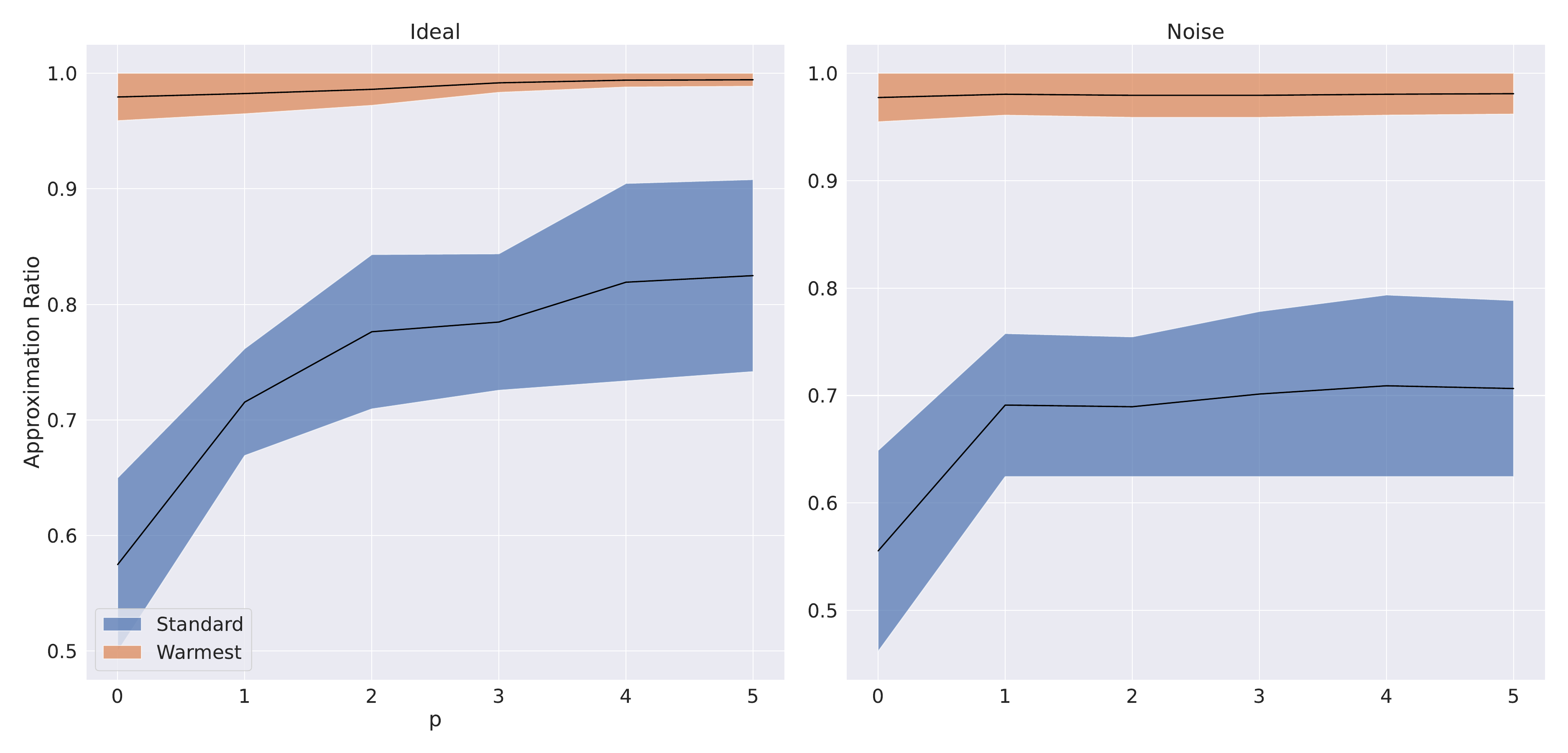}
    \caption{\footnotesize Comparison between standard QAOA mixer to using BM-MC$_2$ warm-starts with custom mixers . We show the noiseless (left) and noisy (right) case. In both cases, the custom mixer significantly outperforms the standard mixer. Shaded regions indicate the distribution of results for 20 randomly chosen 8 node graphs with positive and negative weights.}
    \label{fig:neg_weights}
\end{figure*}

\begin{figure}[!t]
    \centering
    \includegraphics[scale=0.5]{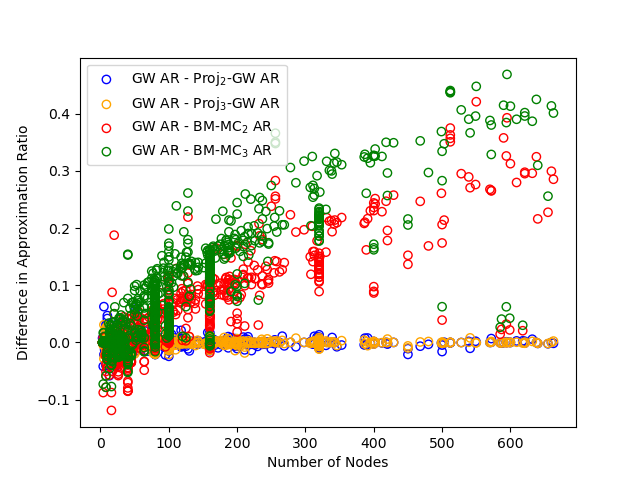}
    \caption{\footnotesize Difference in approximation ratio between $n$-\dimensional GW hyperplane rounding and hyperplane rounding of various warm-start initializations ($k$-\dimensional projected GW SDP solutions (Proj$_k$-GW) and approximate BM-MC$_k$ solutions for $k=2,3$) as the number of nodes varies. For each instance and each \dimension $k$, we obtained 5 random projections and 5 approximate BM-MC$_k$ solutions, and then kept the best one (of 5) in regards to the BM-MC$_k$ objective. Each circle in the figure corresponds to an instance from (a subset of) the MQLib library \cite{DGS18}; see Appendix \ref{sec:GW_BMMC_Scaling} for details. }
    \label{fig:GW_BMMC_Scaling}
\end{figure}

\begin{figure*}[htbp]
    \centering
    \includegraphics[scale=0.5]{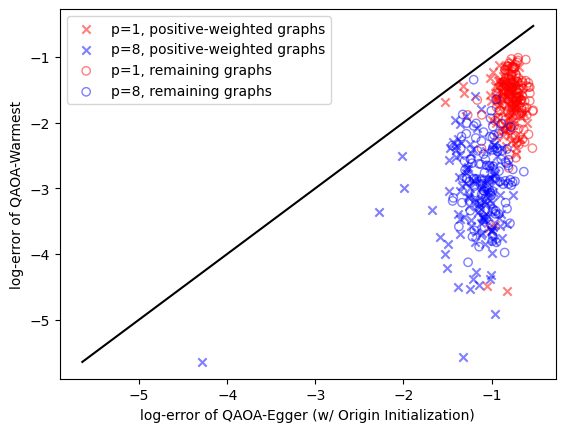}\includegraphics[scale=0.5]{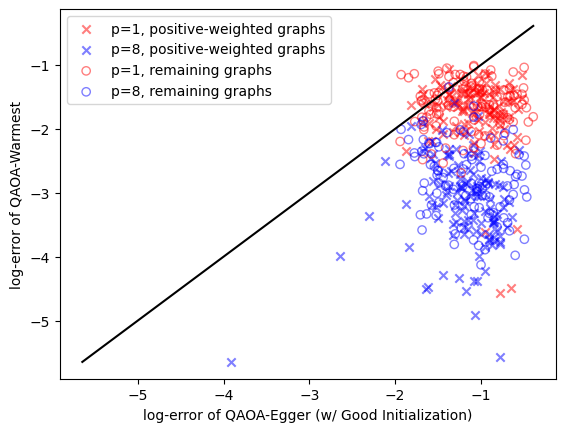}
    \caption{\footnotesize For both plots, we compare the log-error  of QAOA-warmest (with BM-MC$_2$ warm-starts) to the variant of QAOA proposed by Egger et al. \cite{egger2020warm} for Max-Cut. For Egger et al.'s approach, we consider two different initializations: initializing the variational parameters to the origin (left) and initializing the parameters in a way that recovers the cut used to initialize the quantum state (right). Each marker in the plot corresponds to a combination of instance (from our graph ensemble $\mathcal{G}$) and circuit depth (either $p=1$ or $p=8$) with the shape of the marker being used to denote if the instance has only positive edge weights or not. Points below the black line correspond to instances where QAOA-warmest performs better than the other algorithm being compared.}
    \label{fig:eggerComparison}
\end{figure*}

\begin{table*}[htbp]
\centering
\resizebox{\linewidth}{!}{%
\begin{tabular}{c}
\begin{tabular}{|c?c|c|c|c|}
\multicolumn{5}{c}{\textbf{Instances Where Depth-8 Standard QAOA Performs Best}}\\
\hline Instance ID & QAOA-warmest & QAOA-warm & Standard QAOA & GW \\
\hline 778 & 0.9550 & 0.8980 & 0.9635 & 0.9504\\
\hline 1820 & 0.9483 & 0.9078 & 0.9508 & 0.9429\\\hline
\end{tabular}

\\\\\\
\begin{tabular}{|c?c|c|c|c?c|}
\multicolumn{6}{c}{\textbf{Instances Where GW Performs Best}}\\
\hline Instance ID & QAOA-warmest & QAOA-warm & Standard QAOA & GW  & QAOA-warmest (modified)\\
\hline 1698 & 0.9899 & 0.9950 & 0.9677 & 0.9968 & 0.9998  \\
\hline 1889 & 0.9867 & 0.9886 & 0.9461 & 0.9899 & 0.9995\\
\hline 2010 & 0.9762 & 0.9797 & 0.9330 & 0.9948 & 0.9992 \\\hline
\end{tabular}
\end{tabular}
}
\caption{\footnotesize These tables reports the approximation ratios achieved for the five instances (amongst those in our instance library $\mathcal{G}$) for which depth-8 QAOA-warmest did not obtain the best approximation ratio when compared to depth-8 QAOA-warm, depth-8 standard QAOA, and GW. The top and bottom tables are for instances in which standard QAOA and GW performed the best respectively. 
The instances in the bottom table have the property that there exists exactly one negative edge weight whose magnitude is much larger than the other edge weights. For the bottom table, in the last column, we also include the approximation ratio for QAOA-warmest in the case where a more suitable vertex-at-top rotation is used; i.e., we take one of the vertices incident to the large-magnitude negative edge and rotate it to the top.}
\label{fig:problematicInstances}
\end{table*}

\section{Discussion}
\label{sec:discussion}
Our experimental results suggest that our QAOA-warmest method combined with initializations obtained by classical means can outperform both the standard QAOA and the Goemans-Williamson algorithm at relatively shallow circuit depths. Conversely, not all initializations on the Bloch sphere are useful; in particular random initializations under-perform compared to classically obtained initializations. Moreover, adversarial initializations could be chosen if one wanted QAOA to perform poorly (i.e. by putting qubits near the poles of the Bloch sphere that correspond to the minimum cut). Overall, finding a suitable initialization is needed in order to see success in QAOA-warmest. In the case of classically-inspired initializations (e.g. Burer-Monteiro Max-Cut relaxations or projected GW SDP solutions) which are (classically) invariant under global rotations, this also includes picking a suitable rotation scheme before embedding the solution into a quantum state.

According to a paper by Farhi, Gamarnik, and Gutmann \cite{FGG20}, QAOA needs to ``see the whole graph" (i.e. have a high enough circuit depth) in order to achieve desireable results. Their results rely on the fact that local changes in the graph (e.g. modifying an edge weight) give uncorrelated results in regards to measured qubits that are sufficiently far away from such a local change. In other words, standard QAOA cannot distinguish between graphs whose local subgraph-structure is identical. It should be noted that the circuit used in QAOA-warmest also suffers from such a locality property; however, if we consider the entirety of the QAOA-warmest procedure, including the preprocessing stage of computing warm-starts, then this procedure can possibly distinguish between graphs with identical local subgraph structure since the initial state is sensitive to the global structure of the graph (when using BM-MC$_k$ relaxations or projected GW SDP solutions). This suggests that certain negative theoretical results seen for standard QAOA may not necessarily hold for QAOA-warmest since the distinguishability arguments used would no longer apply.


The approximation guarantees for our warm-starts at $p=0$ and convergence to Max-Cut (under adiabatic limit) combined with superior empirical performance provide strong evidence for quantum advantage of this approach at low circuit depths compared to existing classical methods, especially the Goemans-Williamson approximation. An interesting open question would be to quantify the approximation bounds obtained by QAOA-warmest for finite circuit depth greater than zero. 

\begin{table}[htbp]
    \centering
    \resizebox{\linewidth}{!}{%
    \begin{tabular}{|c|c|c|c|c|}
    \Xhline{3\arrayrulewidth}
            Number & Cores & & & \\
          of  &  Per & CPU & RAM & Speed \\ 
          Servers & Server & & & \\
          \Xhline{3\arrayrulewidth}
          1& 12 & E5-2630 & 128 GB & 2.30 GHz\\ \hline
          1 & 32 & Opteron 6274 & 128 GB & 2.20 GHz\\ \hline
          4 & 12 & E5-2630 & 128 GB & 2.30 GHz\\ \hline
          2 & 12 & X5660 & 72 GB & 2.80 GHz\\ \hline
          2 & 12 & X5660 & 148 GB & 2.80 GHz\\ \hline
          4 & 12 & E5645 & 96 GB & 2.40 GHz \\ \Xhline{3\arrayrulewidth}
    \end{tabular}
    }
    \caption{\footnotesize This table details the specifications of the server groups in the high performance computing cluster at Georgia Institute of Technology that were used in our numerical experiments.}
    \label{tab:cluster}
\end{table}

\section{Methods}
\label{sec:methods}
Numerical simulations were performed using both custom and pre-packaged codes in the Tensorflow \cite{tensorflow2015-whitepaper} and Qiskit \cite{qiskit} software packages. Numerical experiments in Section \ref{sec:QAOAWarmestCompare} were performed on the high performance computing cluster at the School of Industrial and Systems Engineering at Georgia Institute of Technology. Jobs were sent to various servers in the cluster as they became available; a listing of the servers and their specifications can be found in Table \ref{tab:cluster}.

Classical optimization was performed using standard optimizers available in {\sc python}, including ADAM \cite{KB14}, L-BFGS-B \cite{bfgsRef}, and COBYLA \cite{cobylaRef}. For hardware results, we first describe the usage of IBM's Guadalupe device along with Qiskit software and the COBYLA optimizer. The Guadalupe device is a 16 qubit superconducting hardware with a heavy hexagonal connectivity. This device typically has average single qubit gate errors of $3.7204\times 10^{-4}$, two-qubit gate errors of $1.075\times 10^{-2}$ and measurement error of $1.776\times 10^{-2}$, representing a quantum device of comparable quality to the state of the art. For the QAOA-Warmest runs shown in Fig~\ref{fig:Guadalupe_comparison}, the run time on hardware can be estimated by the {\bf{schedule}} method of Qiskit for each corresponding circuit. The run times for $p=0,1,2$ are $7182~ \textrm{ns},14968~\textrm{ns},17379~\textrm{ns}$, respectively. Standard QAOA runs have comparable run times as they differ only by single qubit gates as compared to QAOA-Warmest. For these hardware runs, we seed the classical optimization process with ideal parameters ($\gamma^*,\beta^*$) found in simulation and perform 20 optimization steps each with 8192 shots on hardware. Noisy simulations using hardware-informed noise models were performed with 200 optimization steps and 3000 shots. These simulations were performed on GTRI's Icehammer cluster using a node with a Xeon-Gold6242R processor with 80 cores and 376 GB of memory. Secondly, for our results on Quantinuum hardware, we only evaluate a single point in parameter space (at $\gamma^*,\beta^*$) at each $p$ depth with 1000 shots. This choice was motivated in order to reduce the cost of hardware runs and do not represent any device or fundamental limitation. The Quantinuum H1-1 20 qubit device reports average single qubit gate errors of $5 \times 10^{-5}$, two qubit gate errors of $3 \times 10^{-3}$ and SPAM errors of $3 \times 10^{-3}$. Noisy simulation of the Quantinuum device were also performed through the cloud, provided by the Quantinuum service.

\section{Data Availability}
\label{sec:dataAvail}
The graph instances used in the numerical simulation are accessible at the following github repository: \url{https://github.com/swati1729/CI-QuBe}. Additional experimental data is available via request from the corresponding author.

\section{Author Contributions}
Reuben Tate developed the theory for custom mixers, with guidance and helpful discussions with the team including Jai Moondra, Bryan Gard, Greg Mohler and Swati Gupta. Swati Gupta introduced the idea of iterative rounding of Goemans-Williamson SDP solution to lower dimensions. Jai Moondra developed theoretical guarantees for projected GW warm-starts with guidance from the team. All authors contributed to the design of the numerical simulations, and these were implemented by Reuben Tate.  Bryan Gard designed and implemented the experiments on both IBM-Q and Quantinuum hardware. All authors contributed to the editing and writing of the manuscript to enable the best presentation of results.

\section{Competing Interests}
The authors declare that there are no competing financial or non-financial interests.

\section*{Acknowledgement}
A part of this work was done when authors R. Tate and S. Gupta were at the Georgia Institute of Technology. This material is based upon work supported by the Defense Advanced Research Projects Agency (DARPA) under Contract No. HR001120C0046. This research used resources of the Oak Ridge Leadership Computing Facility at the Oak Ridge National Laboratory, which is supported by the Office of Science of the U.S. Department of Energy under Contract No. DE-AC05-00OR22725. We acknowledge the use of IBM Quantum services for this work. The views expressed are those of the authors and do not reflect the official policy or position of IBM or the IBM Quantum team. The authors would also like to thank Hassan Mortagy for his careful comments and feedback on an initial version of this work.

\bibliographystyle{quantum}
\bibliography{references}

\appendix

\section{Proofs}
\label{sec:proofs}

\paragraph{Proof of Proposition} \ref{thm:convergence}
\begin{proof}
Let $\ket{s_0} = \bigotimes_{j=1}^n \ket{s_{0,j}}$ with $\ket{s_{0,j}} = \cos(\theta_j/2)\ket{0}+e^{i\phi_j}\sin(\theta_j/2)\ket{1}$ be an arbitrary separable initial state. As a consequence of Proposition \ref{thm:convergenceSpecial}, it suffices to show that QAOA-warmest with this initial state $\ket{s_0}$ yields the same expected cut value (at the same variational parameters) as another separable initial state $\ket{s_0}'$ where each qubit of $\ket{s_0}'$ lies in the $xz$-plane of the Bloch sphere with positive $x$-coordinate.

We consider the state $\ket{s_0}' = \bigotimes_{j=1}^n \ket{s_{0,j}}'$ where $\ket{s_{0,j}}' = \cos(\theta_j/2)\ket{0}+\sin(\theta_j/2)\ket{1}$. Geometrically, going from $\ket{s_0}$ to $\ket{s_0}'$ has the effect of dropping the phase for all qubits so that they lie in the $xz$-plane of the Bloch sphere with positive $x$-coordinate (assuming that none of the qubits are at the poles).

It suffices to show that we can drop the phase for single qubit of $\ket{s_0}$ (say qubit $k$) without changing the expected cut value; the argument can then be easily repeated for the remaining qubits to show that $\ket{s_0}$ and $\ket{s_0}'$ yield identical expected cut values. In this case, we consider the initial state  $\ket{\widehat{s_0}} = \bigotimes_{j=1}^n \ket{\widehat{s_{0,j}}}$ where $\ket{\widehat{s_{0,k}}} = \ket{s_{0,k}}'$ and $\ket{\widehat{s_{0,j}}} = \ket{s_{0,j}}$ for $j \neq k$ (i.e. only the position of qubit $k$ is modified). Letting $R_{x,k}(\theta),R_{y,k}(\theta),R_{z,k}(\theta)$ represent the standard rotation operator of the $k$th qubit (about axes $x,y,z$ respectively) about the Bloch sphere by angle $\theta$, we can also write
\begin{equation}
\label{eq:zeroPhaseModification}
    \ket{\widehat{s_0}} = R_{z,k}(-\phi_k)\ket{s_0},
\end{equation}
i.e., $\ket{\widehat{s_0}}$ can be obtained from $\ket{s_0}$ by rotating around the $z$-axis (of the Bloch sphere) by the appropriate amount.

Let $H_B$ and $\widehat{H_B}$ be the corresponding custom mixers for $\ket{s_0}$ and $\ket{\widehat{s_0}}$ respectively. Let $U_B(\beta_\ell) = \exp(-i\beta_\ell H_B)$ and $\widehat{U_B}(\beta_\ell) = \exp(-i\beta_\ell \widehat{H_B})$. For convenience, let $U_C(\gamma_\ell) = \exp(-i\gamma_\ell H_C)$. We can write $U_B(\beta_\ell) = U_{B,\neq k}(\beta_\ell)U_{B,k}(\beta_\ell)$ where $U_{B,k}(\beta_\ell)$ is the portion of $U_B(\beta_\ell)$ that acts on qubit $k$ and $U_{B,\neq k}(\beta_\ell)$ is the portion that acts on the remaining qubits; we can similarly write $\widehat{U_B}(\beta_\ell) = U_{B,\neq k}(\beta_\ell)\widehat{U_{B,k}}(\beta_\ell)$ (the part of the mixer that does not affect the $k$th qubit remains the same). Geometrically, the operation $U_{B,k}(\beta_\ell)$ corresponds to rotating qubit $k$ around its original position on the Bloch sphere by angle $2\beta_\ell$ so,
\begin{align*}U_{B,k} = & \Big[R_{z,k}(\phi_k)R_{y,k}(\theta_j)R_{z,k}(2\beta_\ell)\\
&R_{y,k}(-\theta_j)R_{z,k}(-\phi_j)\Big].\end{align*}
The equation above yields the following key relation between $U_{B,k}$ and $\widehat{U_{B,k}}$:
\begin{align}
\label{eq:zeroPhaseMixingRelation}
    &R_{z,k}(-\phi_k)U_{B,k}R_{z,k}(\phi_k)\nonumber \\
    = &R_{y,k}(\theta_j)R_{z,k}(2\beta_\ell)R_{y,k}(-\theta_j) = \widehat{U_{B,k}}.
\end{align}

For convenience, we will let 
$$U(\mathbf{\gamma,\beta}) = \prod_{\ell=1}^p \Big[ U_B(\beta_\ell)U_C(\gamma_\ell) \Big],$$
and
$$\widehat{U}(\mathbf{\gamma,\beta}) = \prod_{\ell=1}^p \Big[ \widehat{U_B}(\beta_\ell)U_C(\gamma_\ell) \Big],$$
i.e., $U_B$ and $\widehat{U_B}$ correspond to the QAOA-warmest circuit (excluding the initial state) for $\ket{s_0}$ and $\ket{\widehat{s_0}}$ respectively.

The claim amounts to showing (up to some global phase) the following:
\begin{align*}
&\bra{s_0}U(\gamma,\beta)^\dagger H_C U(\gamma,\beta)\ket{s_0} \\
=& \bra{\widehat{s_0}}\widehat{U}(\gamma,\beta)^\dagger H_C \widehat{U}(\gamma,\beta)\ket{\widehat{s_0}},
\end{align*}

for any circuit depth $p$ and any variational parameters $\gamma = (\gamma_1,\dots,\gamma_p)$ and $\beta = (\beta_1,\dots,\beta_p)$; and in particular, QAOA-warmest gives the same expected cut value for both $\ket{s_0}$ and $\ket{\widehat{s_0}}$.
 
First we observe that,
\begin{align*}
    & U(\gamma,\beta)\ket{s_0}\\
    =&\prod_{\ell=1}^p \Big[ U_B(\beta_\ell)U_C(\gamma_\ell) \Big] \ket{s_0}\\
    =&\prod_{\ell=1}^p \Big[ U_{B,\neq k}(\beta_\ell)U_{B,k}(\beta_\ell)U_C(\gamma_\ell) \Big] \ket{s_0}\\
    =& \prod_{\ell=1}^p \Big[ U_{B,\neq k}(\beta_\ell)R_{z,k}(\phi_k)\\
    & R_{z,k}(-\phi_k)U_{B,k}(\beta_\ell)R_{z,k}(\phi_k) \\
    & R_{z,k}(-\phi_k)U_C(\gamma_\ell) \Big] \ket{s_0} \\
    =& \prod_{\ell=1}^p \Big[ U_{B,\neq k}(\beta_\ell) R_{z,k}(\phi_k)\widehat{U_{B,k}}(\beta_\ell)\\
    & R_{z,k}(-\phi_k)U_C(\gamma_\ell) \Big] \ket{s_0} \tag{by Equation \ref{eq:zeroPhaseMixingRelation}}
\end{align*}

\begin{align*}
    =& \prod_{\ell=1}^p \Big[ R_{z,k}(\phi_k)U_{B,\neq k}(\beta_\ell)\widehat{U_{B,k}}(\beta_\ell)\\
    &U_C(\gamma_\ell)R_{z,k}(-\phi_k) \Big] \ket{s_0} \tag{commutativity}\\
    =& \prod_{\ell=1}^p \Big[ R_{z,k}(\phi_k)\widehat{U_{B}}(\beta_\ell)U_C(\gamma_\ell)R_{z,k}(-\phi_k) \Big] \ket{s_0} \tag{combine $\widehat{U_{B,k}}$ and $U_{B,\neq k}$}\\
    =& R_{z,k}(\phi_k) \prod_{\ell=1}^p \Big[ \widehat{U_B}(\beta_\ell)U_C(\gamma_\ell) \Big] R_{z,k}(-\phi_k)\ket{s_0} \tag{telescoping}\\
    =& R_{z,k}(\phi_k) \prod_{\ell=1}^p \Big[ \widehat{U_B}(\beta_\ell)U_C(\gamma_\ell) \Big] \ket{\widehat{s_0}} \tag{by Equation \ref{eq:zeroPhaseModification}}\\
    =& R_{z,k}(\phi_k) \widehat{U}(\gamma,\beta)\ket{\widehat{s_0}}.
\end{align*}

We now finally show that QAOA-warmest initialized with $\ket{s_0}$ and $\ket{\widehat{s_0}}$ yield the same value; in particular the extraneous $R_z(\phi_k)$ term from the previous calculations will not effect the measurement due to commutativity with the cost Hamiltonian:

\begin{align*}
& \bra{s_0}U(\gamma,\beta)^\dagger H_C U(\gamma,\beta)\ket{s_0} \\
= & \bra{\widehat{s_0}}\widehat{U}(\gamma,\beta)^\dagger R_z(\phi_k)^\dagger H_C R_z(\phi_k)\widehat{U}(\gamma,\beta)\ket{\widehat{s_0}} \\
= & \bra{\widehat{s_0}}\widehat{U}(\gamma,\beta)^\dagger H_C \widehat{U}(\gamma,\beta)\ket{\widehat{s_0}}, \\
\end{align*}

where the last equality follows since $H_C$ commutes with  $R_z(\phi_k)$. This completes the proof.

\end{proof}

\paragraph{Proof of Theorem \ref{thm: subspace-hyperplane-rounding}.}\label{sec:sdp-warm-start-proof}

\smallskip

Given a subspace $A$ of $\R^n$, if $A = \spn(v)$ for some unit vector $v \in \R$, we abuse the notation and denote $\Pi_v(u) = \Pi_{\spn(v)}(u)$ and $\Lambda_v(u) = \Lambda_{\spn(v)}(u)$. We need two lemmas before we prove the theorem.

\begin{lemma}\label{lem: unit-scale-projections-can-be-composed}
    Let $u, v$ be unit vectors in $\R^n$ and let $A$ denote a linear subspace of $\R^n$ of dimension $k$ such that $v \in A$. If $\Pi_A(u) \neq 0$ and $\Pi_v(u) \neq 0$, then
    \[
        \Lambda_v(u) = \Lambda_v\big(\Lambda_A(u)\big).
    \]
    
    That is, unit-scale projection of $u$ on $v$ is equivalent to first unit-scale projecting $u$ to $A$ and projecting this projection $\Lambda_A(u)$ on $v$.
\end{lemma}

\begin{proof}
    Let $\{v_1, \ldots, v_k\}$ be an orthonormal basis for $A$. Let $\alpha_i = u^\top v_i$. Then
    \[
        \Pi_A(u) = \sum_{i \in [k]} \alpha_i v_i, \quad \Lambda_A(u) = \frac{\sum_{i \in [k]} \alpha_i v_i}{\sqrt{\sum_{i \in [k]} \alpha_i^2}}
    \]
    Since $v \in A$, write $v = \sum_{i \in [k]} \beta_i v_i$. Then we have
    \[
        \Pi_v\big(\Lambda_A(u)\big) = \frac{\sum_{i \in [k]} \alpha_i \beta_i}{\sqrt{\sum_{i \in [k]} \alpha_i^2}} \: v, 
    \]
    so that
    \begin{align*}
        \Lambda_v\big(\Lambda_A(u)\big) &= \frac{\Pi_v\big(\Lambda_A(u)\big)}{\| \Pi_v\big(\Lambda_A(u)\big)\|_2} \\
        &= \frac{\sum_{i \in [k]} \alpha_i \beta_i}{\left| \sum_{i \in [k]} \alpha_i \beta_i\right|} \: v \\ 
        &= \frac{\Pi_v(u)}{\|\Pi_v(u)\|_2} \: v \\
        &= \Lambda_v(u).
    \end{align*}
\end{proof}

Note that the above lemma is deterministic statement; we have not used any randomness so far.

Let us consider what happens if we select a linear subspace $A$ of $\R^n$ of dimension $k$ uniformly randomly from $\R^n$ (one way to ensure it is chosen uniformly randomly is to select unit vectors $v_i \in R^n, i \in [k]$ recursively so that $v_i$ is chosen uniformly randomly in the space orthogonal to $v_1, \ldots, v_{i - 1}$). Once we have $A$, let us select a vector $v \in A$ uniformly randomly again. Is this equivalent to choosing a vector $v \in \R^n$ uniformly randomly? By symmetry, it is, since the former experiment is not biased in favor of any direction. We omit the formal proof and state it as a lemma here:

\begin{lemma}\label{lem: choosing-random-vector-can-be-composed}
    Let $E$ denote the experiment of choosing a unit vector $v$ chosen uniformly randomly from $\R^n$. Let $E'$ denote the experiment of choosing a linear subspace $A$ of dimension $k$ uniformly randomly from $\R^n$, and then choosing a unit vector $v'$ uniformly randomly from $A$. Then $E' = E$, i.e., they correspond to the same probability space.
\end{lemma}

We are ready for the proof of Theorem \ref{thm: subspace-hyperplane-rounding}.

\begin{proof}
    Let $U_n = \{v \in \R^n: \|v\|_2 = 1\}$ be the set of unit vectors in $\R^n$. Recall that for a given probability space, a random variable is a real-valued function on the sample space, or that $X, Y: U_n \to \R$. From Lemma \ref{lem: choosing-random-vector-can-be-composed}, the two experiments  correspond to the same probability space. Therefore, it is enough to show that $X(v) = Y(v)$ for all $v \in U_n$.

    One key observation is that rounding on a random hyperplane is equivalent to unit-scale projecting to a uniformly random vector $v$: indeed, let $v$ be the vector normal to the uniform hyperplane, then any unit vector $u$ is rounded to $1$ if $u \cdot v > 0$ and to $-1$ if $u \cdot v < 0$. That is, $u$ is rounded to $\frac{u \cdot v}{|u \cdot v|} = \Lambda_v(u)^\top v$.
    
    Therefore, 
    \begin{align*}
        X(v) = \!\!\!\!\! \sum_{ij \in E(G)} w_{ij} \: \mathbf{1} \Big[\Lambda_v(u_i)^\top v \cdot \Lambda_v(u_j)^\top v < 0 \Big].
    \end{align*}
    
    Similarly, for a given $A$ such that $v \in A$, we have $Y(v)$ is equal to:
    \begin{align*}
        \sum_{ij \in E(G)} & w_{ij} \mathbf{1}\Big[\Lambda_v\big(\Lambda_A(u_i)\big)^\top v \\
        &\cdot \Lambda_v\big(\Lambda_A(u_j)\big)^\top v < 0 \Big].
    \end{align*}
    
    Since $\dim(A) = k$ is a constant and $A$ is chosen uniformly randomly, $\Pi_A(u_i) \neq 0$ for each $i$ with probability $1$. From Lemma \ref{lem: unit-scale-projections-can-be-composed}, we have $\Lambda_v\big(\Lambda_A(u)\big) = \Lambda_v(u)$ for all unit vectors $u$ and for all $A$.
    
    Therefore, we have,
    \begin{align*}
        &Y(v) \\
        &= \sum_{ij \in E(G)} w_{ij} \: \mathbf{1}\Big[\Lambda_v(u_i)^\top v \cdot \Lambda_v(u_j)^\top v < 0 \Big] \\
        &= X(v).
    \end{align*}
    
    Since $X$ and $Y$ have the same distribution, the same approximation guarantee holds for $X, Y$. This proves part $1$ of the theorem.
    
    We prove part $2$ next. Let $C$ denote the maximum cut value on graph $G$, and denote $\alpha = 0.878$ for convenience. Then, part $1$ shows that $\E Y \ge \alpha C$. We first show that $\Pr\left(Y > (1 - \epsilon)\E Y\right) \ge \alpha \epsilon $ using Markov inequality:
    \begin{align*}
        &\Pr\left(Y > (1 - \epsilon)\E Y\right)\\
        &= 1 - \Pr\left(Y \le (1 - \epsilon)\E Y\right) \\
        &= 1 - \Pr\left(C - Y \ge C - (1 - \epsilon) \E Y\right) \\
        &\ge 1 - \frac{\E\left(C - Y\right)}{C - (1 -  \epsilon) \E Y} \\
        &= \frac{\epsilon \E Y}{C - (1 -  \epsilon) \E Y} \\
        &\ge \frac{\epsilon \E Y}{\frac{\E Y}{\alpha} - (1 -  \epsilon) \E Y} \\
        &= \frac{\alpha \epsilon}{1 - \alpha(1 - \epsilon)} \\
        &\ge \alpha \epsilon.
    \end{align*}
    
    Suppose that $\frac{\log n}{\epsilon}$ independent cuts are produced by applying the two-step rounding procedure $\log n$ times. Then the probability that all of these cuts have value less than $(1 - \epsilon)\E Y$ is at most
    \[
        (1 - \alpha \epsilon)^{\frac{\log n}{\epsilon}} \le \left(e^{-\alpha \epsilon}\right)^{\frac{\log n }{\epsilon}} = \frac{1}{n^{\alpha}},
    \]
    where we have used the standard inequality $\exp(-x) \ge 1 - x$. Since $\alpha \in (0, 1)$, this probability goes to $0$ as $n$ goes to $\infty$.
    
    We prove the second claim of part $2$. Given a $k$-dimensional subspace $A$ of $\R^n$, let $w_A$ denote the average cut value after Goemans-Williamson hyperplane rounding is performed on $A$, i.e.
    \[
        w_A = \frac{\int_{v \in U_A} (\text{cut value along } v) \: dv }{\int_{v \in U_A} \: dv},
    \]
    where $U_A$ is the set of all unit vectors in $A$. Notice that
    \[
        \E Y = \frac{\int_A w_A \: dA}{\int_A \: dA}.
    \]
    
    For the first step of the two-step rounding procedure (i.e., the step selecting a random subspace of dimension $k$), let $Z$ denote the random variable that takes value $w_A$ when subspace $A$ is selected. We need to show that for $\frac{\log n}{\epsilon}$ i.i.d. random variables $Z_1, \ldots, Z_{\frac{\log n}{\epsilon}}$, there is at least some $Z_i$ such that $Z_i \ge (1 - \epsilon) \E Y$. Since the random subspace is selected uniformly randomly, we have that
    \[
        \E Z = \frac{\int_A w_A \: dA}{\int_A \: dA} = \E Y.
    \]
    A similar Markov inequality analysis on $Z$ then gives the result.
    
\end{proof}

\medskip

\paragraph{Proof of Corollary \ref{thm:preserveHyperplaneRatioQuantum}}

\begin{proof}
	Let $F_0^\prime = F_0^\prime(\gamma,\beta)$ be the expected value of {\sc Max-Cut} obtained by quantum sampling (i.e., QAOA for $p=0$). Then,
	\begin{align*}
		&\frac{F_0^\prime}{\text{\mc}(G)}\\
		&\ge \kappa \cdot \frac{F_0^\prime}{\text{HP}(x)}
		\tag{$\text{since \text{HP}}(x) \ge \kappa \text{\mc}(G)$} \\
		&\ge \kappa \min_{(i,j) \in E} \frac{w_{ij}\mathbb{E}[\mathbf{1}[i \text{ \& } j \text{ have different spins}]]}{\frac{w_{ij}}{\pi}\arccos(x_i \cdot x_j)} \tag{{$ \frac{a+c}{b+d} \geq \min(\frac{a}{b}, \frac{c}{d})$ for $a,b,c,d \geq 0$}}\\
      &= \kappa \min_{(i,j) \in E} \frac{\mathbb{E}[\mathbf{1}[i \text{ \& } j \text{ have different spins}]]}{\frac{1}{\pi}\arccos(x_i \cdot x_j)}.\\
	\end{align*}
	
	For $i,j \in [n]$, let $\theta_{ij}$ denote the angle between $x_i$ and $x_j$. We can write $\mathbb{E}[\mathbf{1}[i \text{ and } j \text{ have different spins}]] = f_k(\theta_{ij})$, that is, this expectation is solely a function of the angle between the adjacent vertices and the \dimension  $k$ considered. In particular, in Theorem 2 of Tate et al. \cite{TFHMG20}, they show for $k \in \{2,3\}$ that
	$$f_k(\theta_{ij}) = \frac{1}{2}\left(1 - \frac{\cos \theta_{ij}}{k}\right).$$

	Additionally, $\arccos(x_i \cdot x_j) = \theta_{ij}$. Using these substitutions, we have
		\begin{align*}
		&\frac{F_0^\prime}{\text{\mc}(G)}\\
		&\ge \kappa \min_{(i,j) \in E} \frac{\mathbb{E}[\mathbf{1}[i \text{ and } j \text{ have different spins}]]}{\frac{1}{\pi}\arccos(x_i \cdot x_j)}\\
		&= \kappa \min_{(i,j) \in E} \frac{f_k(\theta_{ij})}{\theta_{ij}/\pi}\\
		&\ge \kappa \min_{\theta \in [0,\pi]} \frac{f_k(\theta)}{\theta/\pi}\\
		&= \kappa \min_{\theta \in [0,\pi]} \frac{ \frac{1}{2}\left(1 - \frac{\cos \theta}{k}\right)}{\theta/\pi}.
	\end{align*}
	
For $k \in \{2,3\}$, it is straightforward to verify that the minimum in the last line above is achieved at $\theta = \pi$ (see Figure \ref{fig:preserveHyperplaneRatioQuantum}) which leads to the ratio $\frac{F_0^\prime}{\text{\mc}(G)}$ being at least $\frac{3}{4}\kappa$ and $\frac{2}{3}\kappa$ respectively for $k=2,3$.
\end{proof}

\paragraph{Proof of Observation} \ref{thm:singleCutAR}

\begin{proof}We prove something more general than what is stated in Proposition \ref{thm:singleCutAR}. Instead of considering a \emph{particular} cut where hyperplane-rounding yields an approximation ratio of $\alpha$, we will instead consider a \emph{random} cut (obtained via hyperplane rounding where the hyperplane is selected uniformly at random). We will show that in expectation (considering both the randomness of the hyperplane rounding and the randomness of the quantum sampling), that the expected cut value is at least $0.878\cos^{2|V|}(\theta^*/2)$.

Let $X$ denote the random variable corresponding to the cut value obtained from the cut obtained by quantum measurement of a perturbed single-cut initialization obtained by GW hyperplane rounding as described in Section \ref{sec:singleCutInitialization}. For any $S \subseteq V$, let $\text{GW}(S)$ denote the event  that the cut $(S, V \setminus S)$ was obtained from the GW hyperplane rounding step. Similarly, let $\text{QM}(S)$ denote the event that quantum measurement of the initial state resulted in the cut $(S, V \setminus S)$.

First, observe that if the cut $(S, V \setminus S)$ is used to initialize the quantum state with regularization angle $\theta^*$, then the probability of quantum measurement getting the same cut is $\cos^{2|V|}(\theta^*/2)$; this is because each vertex independently has probability $\cos^2(\theta^*/2)$ of remaining on the same side of the cut used to initialize the quantum state. Using this fact, we find that,

\begin{align*}
    & \E[X]\\
    =& \sum_{S \subseteq V} \E[X \mid \text{GW}(S)]\Pr(\text{GW}(S)) \\
    =& \sum_{S \subseteq V}\Bigg(\sum_{T \subseteq V}\E[X \mid \text{GW}(S), \text{QM}(T)]\\
    & \cdot \Pr(\text{QM}(T) \mid \text{GW}(S)) \cdot \Pr(\text{GW}(S))\Bigg)\\
    \geq& \sum_{S \subseteq V}\Bigg(\E[X \mid \text{GW}(S), \text{QM}(S)]\\
    & \cdot \Pr(\text{QM}(S) \mid \text{GW}(S)) \cdot \Pr(\text{GW}(S))\Bigg)\\
    =& \sum_{S \subseteq V}\Bigg(\text{cut}(S) \cdot \cos^{2|V|}(\theta^*/2) \cdot \Pr(\text{GW}(S))\Bigg)\\
    =& \cos^{2|V|}(\theta^*/2) \sum_{S \subseteq V}\Bigg(\text{cut}(S)  \Pr(\text{GW}(S))\Bigg)\\
    \geq & \cos^{2|V|}(\theta^*/2) \cdot 0.878 \text{Max-Cut}(G),\\
\end{align*}
and thus $\frac{\E[X]}{\text{Max-Cut}(G)} \geq 0.878\cos^{2|V|}(\theta^*/2)$ as desired. In the above formulas, we used the fact that $\E[X \mid \text{GW}(S), \text{QM}(S)] = \text{cut}(S)$ and that the sum $\sum_{S \subseteq V}(\text{cut}(S)  \Pr(\text{GW}(S)))$ is simply the expected cut value of the GW algorithm, which we know is at least 0.878 of the optimal cut value for graphs with non-negative weights \cite{GW95}.
\end{proof}

\section{Detailed Description of Custom Mixer}
\label{sec:customMixersConstruction}

Consider any separable state $\ket{s_0}$ on $n$ qubits on the Bloch sphere, i.e., $\ket{s_0}$ can be written in the form: 
$$\ket{s_0} = \bigotimes_{j=1}^n \ket{s_{0,j}},$$
where for $j=1,\dots,n$,
$$\ket{s_{0,j}} =  \cos(\theta_j/2)\ket{0}+e^{i\phi_j}\sin(\theta_j/2)\ket{1}.$$
Here, $\theta_j$ and $\phi_j$ can be interpreted as the polar and azimuthal angle respectively of the $j$th qubit on the Bloch sphere. 
The position of the $j$th qubit on the Bloch sphere can also be described in Cartesian coordinates $\hat{n}_j = (x_j,y_j,z_j)$ via the following transformation from spherical to Cartesian coordinates:
\begin{align*}
    x_j &= \sin \theta_j \cos \phi_j, \\
    y_j &= \sin \theta_j \sin \phi_j, \\
    z_j &= \cos \theta_j.
\end{align*}

Recall (from Section \ref{sec:customMixers}) the custom mixing Hamiltonian $H_B$ is then constructed as follows:
$$H_B =\bigoplus_{j=1}^n H_{B,j},$$
where
$H_{B,j} = x_j\sigma^x + y_j\sigma^y + z_j\sigma^z$. 
To develop a geometrical understanding of the custom mixer, consider the operator $R_{\hat{n},j}(\alpha)$ that rotates the $j$th qubit by angle $\alpha$ about the $\hat{n}$-axis for some unit vector $\hat{n} = (x,y,z)$; such as operation can be written as:
$$R_{\hat{n},j}(\alpha) = \exp(-i\frac{\alpha}{2} (x\sigma^x_j+y\sigma^y_j+z\sigma^z_j)).$$ Recall that for the $k$th of the $p$ stages of the QAOA circuit (where $p$ is the circuit depth), one applies the unitary operator $e^{-i\beta_k H_B}$ with $\beta_k$ being a variational parameter (to be optimized); this operator, $e^{-i\beta_k H_B}$, can be written as $\prod_{j=1}^n R_{\hat{n}_j, {j}}(2\beta_k)$, i.e., in the $k$th stage of the QAOA circuit, one independently rotates the $j$th qubit about the axis determined by its original position by angle $2\beta_k$.

\section{Convergence of Custom Mixers in xz-plane}
\label{sec:xz_convergence}

In order to prove Proposition \ref{thm:convergenceSpecial} using the adiabatic theorem, we need to show that (1) $\|s_0\rangle$ is indeed the highest-energy eigenstate of the corresponding custom mixer $H_B$, and (2) the difference between the largest and the second-largest eigenvalues of $H(t) = (1 - t/T)H_B + (t/T) H_C$ is strictly positive. We divide this section into two parts to prove these two statements.


\subsection{Eigenstates of Custom Mixers}
We show first that $|s_0\rangle$ is the highest energy eigenstate of the corresponding custom mixer for a single qubit, and then generalize to the Kronecker sums of matrices (Proposition \ref{prop:eigenKroneckerSum}).

\begin{lemma}
\label{thm:groundState} Let $\ket{s} = \cos(\theta/2)\ket{0} + e^{i\phi}\sin(\theta/2)\ket{1}$ be a single-qubit quantum state and let $\hat{n} = (x,y,z)$ be the Cartesian coordinates of that qubit on the Bloch sphere. Let $U =  x\sigma^x +y\sigma^y +z\sigma^z$. Then $\ket{s}$ is the most-excited eigenstate of $U$.
\end{lemma}

\begin{proof}
We have the following relationship between the Cartesian and spherical coordinates: $x = \cos \phi \sin \theta, y = \sin \phi \sin \theta, z = \cos \theta$. Thus, the matrix $U = x\sigma^x+y\sigma^y+z\sigma^z$ is given by $$U = \begin{bmatrix} \cos \theta & \sin \theta e^{-i\phi} \\ \sin \theta e^{i\phi} & -\cos \theta\end{bmatrix}.$$ One can show that the matrix can be diagonalized as $U = PDP^{-1}$ where $P = [v_1 \  v_2], D = \text{diag}(1,-1), v_1 = \cos(\theta/2) \ket{0} + \sin(\theta/2) e^{i\phi}\ket{1}, v_2 = -\sin(\theta/2) \ket{0} + \cos(\theta/2) e^{i\phi}\ket{1}$. Thus, $v_1 =\cos(\theta/2) \ket{0} + \sin(\theta/2) e^{i\phi}\ket{1}$ is the highest-energy eigenstate of $U$.
\end{proof}

We can then formulate the most-excited eigenstate of $U$ using the following relation between eigenvalues of matrices involved in a Kronecker sum and the resultant matrix.

\begin{theorem}
\label{thm:eigenKroneckerSum}
(Theorem 13.16 in \cite{laub}) Let $A \in \mathbb{C}^{n \times n}$  have eigenvalues $\lambda_1,\dots,\lambda_n$ and let $B \in \mathbb{C}^{m \times m}$ have eigenvalues $\mu_1,\dots,\mu_m$. Then the Kronecker sum $A \oplus B$ has $mn$ eigenvalues given by $\{\lambda_i + \mu_j : i \in [n], j \in [m]\}$. Moreover, if $x_1, \dots, x_p$ ($p \leq n$) are linearly independent eigenvectors of $A$ corresponding to $\lambda_1,\dots,\lambda_n$ and $z_1,\dots, z_q$ ($q \leq m$) are linearly independent eigenvectors of $B$ corresponding to $\mu_1,\dots, \mu_q$, then, for all $i \in [p]$ and $j \in [q]$, we have that $x_i \otimes z_j$ are linearly independent eigenvectors of $A \oplus B$ corresponding to $\lambda_i + \mu_j$.
\end{theorem}

By applying Theorem \ref{thm:eigenKroneckerSum} with the summands $H_{B,j}$ of the mixing Hamiltonian $H_B$, we get the desired result as shown in Proposition \ref{prop:eigenKroneckerSum}.

\begin{prop}
\label{prop:eigenKroneckerSum}
Let $\ket{s_0}$ be any separable initial state and let $H_B$ be its corresponding custom mixer. Then $\ket{s_0}$ is the highest-energy eigenstate of $H_B$.
\end{prop}
\begin{proof}
    Suppose for each $j=1,\dots,n$ we have a matrix $A_j$ with real eigenvalues and suppose the largest eigenvalue of $A_j$ is $\lambda_j$ with corresponding eigenvector $v_j$. 
    As a consequence of Theorem \ref{thm:eigenKroneckerSum}, we have that the largest eigenvalue of $\bigoplus_{j=1}^n A_j$ is $\sum_{j=1}^n \lambda_j$ with one corresponding eigenvector being $\bigotimes_{j=1}^n v_j$.  Letting $A_j = H_{B,j}$ and $v_j = \ket{s_{0,j}}$ and applying Lemma \ref{thm:groundState}, we see that $\ket{s_0} = \bigotimes_{j=1}^n \ket{s_{0,j}}$ is a highest-energy eigenstate for $H_B = \bigoplus_{j=1}^n H_{B,j}$.
\end{proof}

\subsection{Eigenvalue Gap For Custom Mixers}

\label{sec:spectralGap}
It is known that if the Quantum Adiabatic Algorithm is run for 
large enough time $T$ with time-varying Hamiltonian $H(t) = (1-t/T)H_B+(t/T)H_C$ starting with the highest-energy eigenstate of $H(0) = H_B$, then one can arrive at the highest-energy eigenstate of $H(T) = H_C$, i.e. the optimal solution, provided that the gap between the largest and second-largest eigenvalue of $H(t)$ is strictly positive for all $t < T$. This translates to finding an optimal solution when running QAOA as we let the circuit depth $p$ tend to infinity. Farhi et al. showed that this eigenvalue gap was strictly positive for standard QAOA \cite{FGG14}, thus guaranteeing convergence to the optimal solution. In particular, they applied the following Perron-Frobenius theorem to irreducible stoquastic\footnote{Stoquastic matrices are square matrices with real entries so that all of the off-diagonal entries are non-negative. Let $A$ be an $n \times n$ square matrix. Construct a directed graph $G_A$ with vertex set $[n]$ where the edge $(i,j)$ is included if and only if $A_{ij} > 0$. If $G_A$ is strongly connected, then we say that $A$ is irreducible. Otherwise, we say that $A$ is reducible.} matrices. 

\begin{theorem}{\cite{Frob12}}
\label{thm:perronFrob}
 Let $A$ be an irreducible matrix whose entries are all real and non-negative. Let $r$ be the spectral radius of $A$, i.e., $r = \max\{ |\lambda|: \lambda \text{ is eigenvalue of } A\}$. Then $r$ is an eigenvalue of $A$ and furthermore, it has algebraic multiplicity of 1.
\end{theorem}

If the eigenvalues of an $n \times n$ matrix $A$ are real (e.g. if $A$ is Hermitian), then its eigenvalues (with multiplicity) can be ordered as $\lambda_1 \geq \cdots \geq \lambda_n$; if $A$ is also irreducible and has real, non-negative entries then Theorem \ref{thm:perronFrob} ensures a gap between the two largest eigenvalues (otherwise, if $\lambda_{1} = \lambda_2$, then the algebraic multiplicity of $\lambda_1$ would be at least 2, contradicting the statement of the theorem.) This observation still holds if we relax the non-negativity condition to allow negative entries along the diagonal as seen in the following lemma.

\begin{lemma}
\label{thm:perronStoquastic}
Let $A$ be an irreducible stoquastic Hermitian matrix. Then the difference between the largest and second-largest eigenvalue of $A$ is strictly positive.
\end{lemma}

\begin{proof}
Since $A$ is stoquastic, then all of the off-diagonal elements are already non-negative; however, the diagonal elements may be negative. Observe that for large enough $k$, we have that $A+kI$ is a matrix with all non-negative entries. Note that $A+kI$ is Hermitian (since $A$ and $I$ are) and thus the eigenvalues of $A+kI$ are real. If we apply the Perron-Frobenius theorem to $A+kI$, one observes that the gap between the largest and second-largest eigenvalue is strictly positive.

One can show that the eigenvalues of $A+kI$ can be obtained by shifting all of the eigenvalues of $A$ by $k$ (i.e. of $\lambda$ is an eigenvalue of $A$, then $\lambda+k$ is an eigenvalue of $A+kI$). Moreover, the multiplicities of these shifted eigenvalues are preserved. Thus, the gap between the largest and second-largest eigenvalue of $A+kI$ (which is strictly positive) is equal to the gap between the largest and second-largest eigenvalue of $A$.
\end{proof}

If the custom mixer $H_B$ has the form $\sum_{j=1}^n (x_j\sigma^x_j+z_j\sigma^z_j)$ with $x_j \in \mathbb{R}^+$ and $z_j \in \mathbb{R}$ for $j=1,\dots,n$, then one can show that $H(t)$ is an irreducible, stoquastic matrix. Thus by Lemma \ref{thm:perronStoquastic}, the eigenvalue gap of $H(t)$ is strictly positive meaning that one can achieve the optimal solution as the circuit depth $p\to \infty$ in QAOA-warmest. Geometrically, this special case corresponds to an initial separable state whose qubits lie in the $xz$-plane on the Bloch sphere with $x > 0$. The stoquasticity and irreducibility of this special case is formalized in the following two propositions respectively. We include their proofs at the end of this section, and go on to prove our main result first.

\begin{prop}
\label{thm:specialcaseStoquastic} Let $n$ be a positive integer. For each $j=1,\dots,n$ let $x_j$ be any non-negative real number and let $z_j$ be any real number. Let $H_B = \sum_{j=1}^n (x_j\sigma^x_j+z_j\sigma^z_j)$ and let $H_C$ be the problem Hamiltonian for QAOA. Then $H(t) = (1-t/T)H_B+(t/T)H_C$ is stoquastic for all $0 \leq t \leq T$.
\end{prop}
\begin{prop}
\label{thm:specialcaseIrr}
Let $n$ be a positive integer. For each $j=1,\dots,n$ let $x_j$ be a positive real number and let $z_j$ be any real real number. Let  $H_B = \sum_{j=1}^n (x_j\sigma^x_j+z_j\sigma^z_j)$ and let $H_C$ be the problem Hamiltonian for QAOA. Then $H(t) = (1-t/T)H_B+(t/T)H_C$ is irreducible for all $0\leq t < T$. 
\end{prop}

We can now prove the convergence for custom mixers of the special form ($\sum_{j=1}^n (x_j\sigma^x_j+z_j\sigma^z_j)$ with $x_j \in \mathbb{R}^+$ and $z_j \in \mathbb{R}$ for $j=1,\dots,n$) and their corresponding initializations {as described in Proposition \ref{thm:convergenceSpecial} in Section \ref{sec:customMixers}.}

\begin{proof}
    By construction, the corresponding custom mixers will have the form $\sum_{j=1}^n (x_j\sigma^x_j+z_j\sigma^z_j)$ with $x_j \in \mathbb{R}^+$ and $z_j \in \mathbb{R}$ for $j=1,\dots,n$. The result then follows from Proposition \ref{prop:eigenKroneckerSum}, Lemma \ref{thm:perronStoquastic} (which is applicable due to Proposition \ref{thm:specialcaseStoquastic} and Proposition \ref{thm:specialcaseIrr}), and the adiabatic theorem.
\end{proof}

Recall that for standard QAOA, we have that $H_B = \sum_{j=1}^n (1 \cdot \sigma^x_j + 0\cdot \sigma^z_j)$. Thus, the fact that standard QAOA converges to the optimal cut as $p \to \infty$ is a special case of Propositions \ref{thm:specialcaseStoquastic} and \ref{thm:specialcaseIrr}. \\

\paragraph{Proof of {Proposition} \ref{thm:specialcaseStoquastic}}

We first prove a technical lemma that is needed in order to prove the proposition.

\begin{lemma}
\label{thm:stoquasticKronecker}
If $A$ and $B$ are $n \times n$ and $m \times m$ stoquastic matrices respectively, then so is $A \oplus B$.
\end{lemma}
\begin{proof}
By definition, $A \oplus B = A \otimes I_m + I_n \otimes B$. Since we know the sum of stoquastic matrices is stoquastic, it suffices to show that $A \otimes I_m$ and $I_n \otimes B$ is stoquastic.

Observe that,
$$A \otimes I_m = \begin{bmatrix} 
A_{11}I_m & \cdots & A_{1n}I_m \\
\vdots & \ddots & \vdots \\
A_{n1}I_m & \cdots & A_{nn}I_m\end{bmatrix}.$$ Note that for $i \neq j$, the $ij$th block in the block matrix above is $A_{ij}I_m$, which contains only non-negative entries since $A_{ij}$ is an off-diagonal element of $A$ and $A$ is stoquastic. Now consider the entries in the $ij$th block where $i=j$. Note that if there is an off-diagonal entry of $A \otimes I_m$ that is part of the $ii$th block, then it is also an off-diagonal entry of that block but all the off-diagonal entries of the $ii$th block ($A_{ii}I$) are zero. Thus, we have shown that every off-diagonal element is non-negative, thus $A \otimes I_m$ is stoquastic.

Next, observe that
$$I_n \otimes B = \begin{bmatrix} 
1B & 0B & \cdots & 0B \\
0B & \ddots & \ddots & \vdots \\
\vdots & \ddots & \ddots & 0B \\
0B & \dots & 0B & 1B\end{bmatrix},$$
which makes it clear that the off-diagonal elements of $I_n \otimes B$ are either 0 or the off-diagonal elements of B which are non-negative (by stoquasticity of $B$) and thus $I_n \otimes B$ is stoquastic.

\end{proof}
We are now ready to prove {Proposition \ref{thm:specialcaseStoquastic}.}
\begin{proof}
By construction $H_C$ (and thus $(t/T)H_C$) is a diagonal matrix (as $\ket{b}$ is an eigenvector of $H_C$ for each $n$-length bitstring $b$). If $H_B$ were stoquastic, then $(1-t/T)H_B$ is stoquastic (as $1-t/T \geq 0$ for $0 \leq t \leq T$) and thus $H(t) = (1-t/T)H_B+(t/T)H_C$ is stoquastic (as adding a diagonal matrix to a stoquastic matrix yields a stoquastic matrix). Thus, it remains to show that $H_B$ is stoquastic.

Let $H_{B,j} =  x_j\sigma^x+z_j\sigma^z$. Expanding $\sigma^x$ and $\sigma^z$, we have that
$$H_{B,j} = \begin{bmatrix} z_j & x_j \\ x_j & -z_j\end{bmatrix},$$
which is clearly stoquastic as we assumed that $x_j \geq 0$. As $H_B = \bigoplus_{j=1}^n H_{B,j}$, the result now follows from Lemma \ref{thm:stoquasticKronecker}.
\end{proof}

\paragraph{Proof of {Proposition} \ref{thm:specialcaseIrr}}
\begin{proof}
First, we recall the definition of irreducible matrix. Let $A$ be an $n \times n$ square matrix. Construct a directed graph $G_A$ with vertex set $[n]$ where the edge $(i,j)$ is included if and only if $A_{ij} > 0$. If $G_A$ is strongly connected, then we say that $A$ is irreducible. Otherwise, we say that $A$ is reducible.

For any square matrix $M$, let $G_M$ be the corresponding directed graph as described above. Observe that $H_C$ (and hence $(t/T)H_C$) is a diagonal matrix, thus, by the definition of irreducibility, the irreducibility of $(1-t/T)H_B+(t/T)H_C$ is the same as $(1-t/T)H_B$. Similarly, scaling a matrix by a positive constant does not affect its irreducibility, so it suffices to prove the irreducibility of $H_B$.

Observe that $\sigma^x,\sigma^z$ are symmetric and thus it is not very difficult to show that $H_B$ is also symmetric. This means, for the purposes of showing irreducibility, $G_{H_B}$ is effectively an undirected graph and we just need to show that it is connected. One can write $H_B$ as $H_B = \bigoplus_{j=1}^n (x_j\sigma^x+z_j\sigma^z)$ where $\bigoplus$ denotes the Kronecker sum. According to \cite{KR05}, this means that $G_{H_B}$ can be written as the Cartesian graph product of the graphs $H_1,H_2,\dots,H_n$ where $H_j = G_{A_j}$ with $A_j = x_j\sigma^x+z_j\sigma^z$. Observe, that each of the $H_j$'s are connected if and only if $x_j \neq 0$ which is true by assumption. Since each of the $H_j$'s are connected, then it is also the case that $G_{H_B}$ is connected as well (see Theorem 1 of \cite{Spacapan08}) which finishes the proof.

\end{proof}

\section{Details on Perturbed Single-Cut Initialization}
\label{sec:singleCutAppendix}

In a perturbed single-cut initialization scheme with cut $(S, V \setminus S)$ and regularization angle $\theta^*$, the quantum state is given by,
$${\ket{s_0} = \bigotimes_{j=1}^n \ket{s_{0,j}},}$$
where,
$$\ket{s_{0,j}} = \begin{cases}
    R_Y(\theta^*)\ket{0}, & j \in S\\
    R_Y(\pi - \theta^*)\ket{0}, & j \notin S
\end{cases}$$
where $R_Y(\theta)$ is a single-qubit rotation about the $y$-axis by angle $\theta$. Geometrically, qubits lie on the $xz$-plane of the Bloch sphere (with $x>0$) so that they lie at an angle $\theta^*$ away from either the north or south pole of the Bloch sphere, depending on which side of the cut $(S, V \setminus S)$ the vertex is on.

In one of their approaches for Max-Cut, Egger et al. \cite{egger2020warm} use a mixer for QAOA that is different than both the standard mixer and the custom mixers described in Section \ref{sec:customMixers}. They show that their mixer has the property that, when a \{perturbed single-cut initialization (based on a cut $(S, V \setminus S)$) with regularization parameter $\theta^* = \pi/3$ is used, that measurement of the depth-1 QAOA with variational parameters $(\gamma_1,\beta_1) = (0,\frac{\pi}{2})$ produces exactly the cut $(S, V \setminus S)$ that was used to initialize the initial quantum state. The drawback is that, with such a mixer proposed by Egger et al., no convergence guarantees are known and experiments suggest that, unlike standard QAOA, the optimal cut value is not achieved in expectation with increased circuit depth.

Cain et al. \cite{CFGRT22} consider the case where $\theta^* = 0$ and the standard mixer is used. They find that such an approach performs very poorly; in particular, no convergence towards the optimal cut is found with increased circuit depth either.

One can also consider using a perturbed single-cut initialization together with the custom mixers proposed in Section \ref{sec:customMixers}; this idea was very briefly explored in the appendices of Egger et al.'s work \cite{egger2020warm}. From Proposition \ref{thm:singleCutAR} and Theorem \ref{thm:convergenceGeneral}, it is clear that this approach (with non-negative weighted graphs) yields an approximation ratio approaching 0.878 for $\theta^* \to 0$ for depth-0 QAOA and that such an approach convergences to the optimal cut with infinite circuit depth.

Recall (Section \ref{sec:GWRelaxations}) that QAOA with a (2-dimensional) projected GW initialization has a depth-0 approximation ratio of 0.658. One may be led to believe that, when custom mixers are used, that a perturbed single-cut initialization (with small regularization angle $\theta^*$) is the better choice due to its (theoretically) better approximation ratio at depth-0. However, as seen empirically in Section \ref{sec:experiments}, this is not the case: when $\theta^*$ is small, the convergence rate of QAOA with single-cut initializations is (empirically) incredibly slow across all instances. For small $\theta^*$, QAOA with custom mixers geometrically performs rotations around axes that are near the poles of the Bloch sphere about the qubits' initial positions; it is possible that this geometric interpretation is responsible for the slow convergence for small $\theta^*$.

 \subsection{Relation to QUBO Approach}
Egger et al. \cite{egger2020warm} consider a warm-start approach (which they call continuous warm-started QAOA) for QUBO's of the form $$\min_{y \in \{0,1\}^n} y^TMy,$$
 where $M \in \mathbb{R}^{n \times n}$ is a real-symmetric matrix.
  They then consider the relaxation $$\min_{y \in [0,1]^n} y^TMy,$$
 i.e. the binary variables are relaxed to lie in the interval $[0,1]$. For certain matrices\footnote{In particular, Egger et al. \cite{egger2020warm} consider matrices of the form $M = N+D$ where $N\in \mathbb{R}^{n \times n}$ is a (symmetric) positive-semidefinite matrix and $D\in \mathbb{R}^{n \times n}$ is a diagonal matrix.} $M$, this yields a convex quadratic program which can be easily solved to (global) optimality \cite{GG06}; in this case, Egger et al. \cite{egger2020warm} find and use the globally optimal solution $y^*$ of the relaxation to produce a separable quantum initial state. We consider this approach in the context of Max-Cut, specifically in the case of graphs with non-negative edge weights.

One can formulate Max-Cut on a graph $G=(V,E)$ with edge weights $w: E \to \mathbb{R}$ as follows \cite{DGS18}. Simply construct the QUBO matrix $M$ by setting $M_{ij} = w_{ij}$ for $i \neq j$ and $M_{ii} = -\sum_{j=1}^n w_{ij}$ for $i \in \{1,\dots,n\}$. If $x^*$ is an optimal solution to Max-Cut (using the formulation in Equation \ref{eqn:maxcutFormulation}), then there is a corresponding $y^*$ that is an optimal solution of the QUBO such that $x_i^* = 2y_i^*-1$ for $i=1,\dots,n$.

Observe that $M = -L$ where $L$ Laplacian matrix of the graph $G$ which is known to be positive-semidefinite (for graphs with non-negative edge weights) \cite{chung1997spectral}. Since $L$ is positive-semidefinite, then the function $f(x) = x^TLx$ is convex in $x$ (as the Hessian of $f$, $\nabla^2f(x) = 2L$, is positive-semidefinite). Thus, $x^TMx = -f(x)$ is generally not convex, and hence solving $\min_{x \in [0,1]^n} x^TMx$ to global optimality (as is done by Egger et al. \cite{egger2020warm}) is non-trivial in the case of Max-Cut. However, we can still consider locally optimal solutions to the relaxation. Observe,
 $$\min_{y \in [0,1]^n} y^TMy = \max_{y \in [0,1]^n} y^TLy,$$
 i.e., the QUBO relaxation amounts to maximizing a convex function over a polytope, in which case, all strictly local maxima lie on the vertices of the polytope.\footnote{To see this, suppose by means of contradiction that $y^*$ was a strict local maximum that did not lie at a vertex of the polytope. Then there exists $z \in \mathbb{R}^n$ such that $y^*-z$ and $y^*+z$ lie in the polytope such that $f(y^*) > f(y^*-z)$ and $f(y^*) > f(y^*+z)$. By convexity,
 $f(y^*) = f(\frac{1}{2}(y^*-z)+\frac{1}{2}(y^*+z)) \leq  \frac{1}{2}f(y^*-z)+\frac{1}{2}f(y^*+z) < \frac{1}{2}f(y^*) + \frac{1}{2}f(y^*) = f(y^*), $ a contradiction.} The vertices of the polytope $[0,1]^n$ correspond to cuts in the graph, thus, using the strictly locally optimal solution to the relaxation of the QUBO corresponding to Max-Cut degenerates to solutions corresponding to a single-cut; this means that, for Max-Cut (with non-negative edge-weights), this QUBO approach is (effectively) a single-cut initialization approach as described in Section \ref{sec:singleCutInitialization}.

\section{Noise Simulations}
\label{sec:app_noise}
In Figure \ref{fig:neg_weights}, we consider 20 instances of Erd\H{o}s-R\'enyi graphs, with edge probabilities of 50\% and 8 nodes. In this case, we also choose random negative and positive edges weights from a uniform distribution defined on $ [-1,1] $. We show results for the ideal, noiseless case as well as the noisy case when 3\% phase noise is present. For this noisy case, we consider only one simple source of noise, phase damping, which is present for every single qubit gate operation. This choice is made for simplicity and the following discussion is used as a example of generic noise channel modelling.

 As an example, we consider the modelling of phase noise using Kraus operators which equivalently can be descried using noise channels. Generically, this noise is due to interactions with the environment which is composed of many subsystems. Each interaction itself is weak, but the result of many such interactions, while not likely to cause energy transitions, does introduce a loss of phase coherence. We can describe this process with a set of (non-unique) Kraus operators\cite{Nielsen2011}, given by,
\[
M_0 =\sqrt{1-q}\mathbb{I},\]
\[M_1 = \sqrt{q}\ket{0}\bra{0}, M_2 = \sqrt{q}\ket{1}\bra{1} ,
\]
where $q$ is the probability of a dephasing event occurring.
These Kraus operators then have the effect on the state evolution as,
\begin{align*}
    S(\rho) =& \sum_{k=0}^2 M_k \rho M_k^\dagger\\
    =&(1-q)\rho+q\ket{0}\bra{0}\rho\ket{0}\bra{0} \\
    &+q\ket{1}\bra{1}\rho\ket{1}\bra{1}.
\end{align*}

If we associate the probability of a dephasing event occuring during a time interval $\Delta t$, then $q = \Gamma \Delta t$ with $\Gamma$ being the characteristic dephasing rate and we can write $n$ applications of the noisy channel, $S(\rho, t)$ in matrix form as,
\[
S(\rho,t) = 
\begin{pmatrix}
     \rho_{00} & e^{-\Gamma t}\rho_{01}\\
     e^{-\Gamma t}\rho_{10} & \rho_{11}
\end{pmatrix} .
\]
We can then see that this dephasing process preserves population as $\rho_{00},\rho_{11}$ are preserved but exponentially suppresses coherences at a rate defined by $\Gamma$, which also defines the dephasing time, $T_2 = 1/\Gamma$.

In addition to modelling phase noise, we also include several other noise models in the results shown in Figures \ref{fig:Karlov_20}, \ref{fig:Guadalupe_comparison}, and \ref{fig:quantinuum} of the main text. The specifications for these noise models are generated from Qiskit's {\bf{NoiseModel.from\_backend} } method while using the $``ibm\_ guadalupe"$ device as the targetted backend. In total, these noisy simulations utilize noise models that incorporate gate error probability of each gate, the gate length of all gates, the $T_1$ and $T_2$ times of all qubits, as well as the readout error probability. Each gate error consists of a {\bf{deploarizing\_error()}} followed by a {\bf{thermal\_relaxation\_error()}} error channel. One can define similar Kraus operators for these noise channels as well~\cite{Nielsen2011}. However, while these are a comprehensive treatment of quantum noise they do not accurately capture crosstalk and other correlated noise sources.

As mentioned in the main text, we apply a SPAM mitigation strategy to our hardware results from IBM. Typically these SPAM errors can vary significantly across qubits on a single device. Specifically, typical readout errors on the $ibmq\_guadalupe$ device range from $[1.07\%,12.95\%]$. Since the overall result is limited by the largest, worst case error, for simplicity, we only mitigate SPAM errors due to the fact that these errors are roughly an order of magnitude larger than gate errors on the same device as well as requiring only two additional circuit runs (independent of circuit width or depth) and do not implement any gate error mitigation strategies, which typically require significantly larger overhead~\cite{pascuzzi_computationally_2022,berg_probabilistic_2022}.
\section{Additional Numerical Results}
\label{sec:additionalExperiments}

\subsection{GW vs BM-MC Warm-Starts}

As described in Section \ref{sec:warmstarts}, we consider two approaches for generating warm-starts: projected GW solutions and locally optimal BM-MC$_k$ solutions, with the former approach having better theoretical guarantees in regard to solution quality. However, numerical simulations displayed in Figure \ref{fig:GWVsBMMC} show both approaches on the instance library $\mathcal{G}$ (graphs with at most 11 nodes) achieve similar expected cut values at depth $p=1$ QAOA-warmest; in particular, the difference in {(instance-specific)} approximation ratio is less than 0.04 for nearly all instances. This similarity in solution quality is even more pronounced at depth $p=8$.


Since both warm-start approaches yield similar results (in numerical simulations) and since the Burer-Monteiro approach scales better in regards to runtime (Section \ref{sec:warmstarts}), {the results in Section \ref{sec:experiments} } assumes that locally optimal BM-MC$_k$ solutions are used to produce the warm-starts for QAOA-warm and QAOA-warmest.

\begin{table*}[htbp]
    \centering
    \begin{tabular}{ccc}
    & depth $p=1$ & depth $p=8$\\\\
    all graphs & 
    \begin{tabular}{|c|c|c|} \hline
    &  vert.  & uniform\\\hline
    $k = 2$ & {\bf 0.9858} & 0.9758 \\ \hline
    $k = 3$ &  0.9854 & 0.9535 \\ \hline
    \end{tabular}
    &
    \begin{tabular}{|c|c|c|} \hline
    &  vert.  & uniform\\\hline
    $k = 2$ & {\bf 0.9988} & 0.9977\\ \hline
    $k = 3$ & 0.9988 & 0.9960\\ \hline
    \end{tabular}
    \\\\
    \begin{tabular}{c}positive-weight \\ graphs\end{tabular}
    &
    \begin{tabular}{|c|c|c|} \hline
    &  vert.  & uniform\\\hline
    $k = 2$ & {\bf 0.9867} & 0.9789 \\ \hline
    $k = 3$ & 0.9864 & 0.9581 \\ \hline
    \end{tabular}
    &
    \begin{tabular}{|c|c|c|} \hline
    &  vert.  & uniform\\\hline
    $k = 2$ & 0.9991 & 0.9990 \\ \hline
    $k = 3$ & {\bf 0.9993} & 0.9981 \\ \hline
    \end{tabular}
    \end{tabular}
    \caption{\footnotesize Multiple tables comparing the average {(instance-specific)} approximation ratio achieved during QAOA-warmest\ when utilizing different combinations of \dimension and rotations during the preprocessing stage. For the top row of tables, these averages were computed using all the graphs in our graph library $\mathcal{G}$ whereas for the bottom row, we restrict our attention to only those graphs in $\mathcal{G}$ with positive edge weights.}
    \label{fig:comparingRanksAndRotation}
\end{table*}

\subsection{Choice of Dimension and Rotation}
\label{sec:comparingRanksAndRotation}
Table \ref{fig:comparingRanksAndRotation} demonstrates the average {(instance-specific)} approximation ratio achieved by QAOA-warmest for various combinations of the \dimension used for BM-MC$_k$ and the rotation scheme applied to the BM-MC$_k$ solution. We find that vertex-at-top rotations perform better than uniform rotations, especially in the context of $k = 3$ solutions. The data is inconclusive in regards to if $k = 2$ or $k = 3$ solutions are better for QAOA-warmest, both are promising. Finally, we remark that for depth-8 QAOA-warmest, any choice of \dimension or rotation scheme gave at least a 0.996 average approximation ratio across the instances.

To give fair comparison against QAOA-warm \cite{TFHMG20} (and also for comparisons with Egger et al.'s approach \cite{egger2020warm}), {the results in Section \ref{sec:experiments}}  assumes that we are using $k = 2$ initializations and vertex-at-top rotations unless otherwise stated, since these were the recommended setting in \cite{TFHMG20}. 

\begin{figure}[htbp]
    \centering
    \includegraphics[scale=0.3]{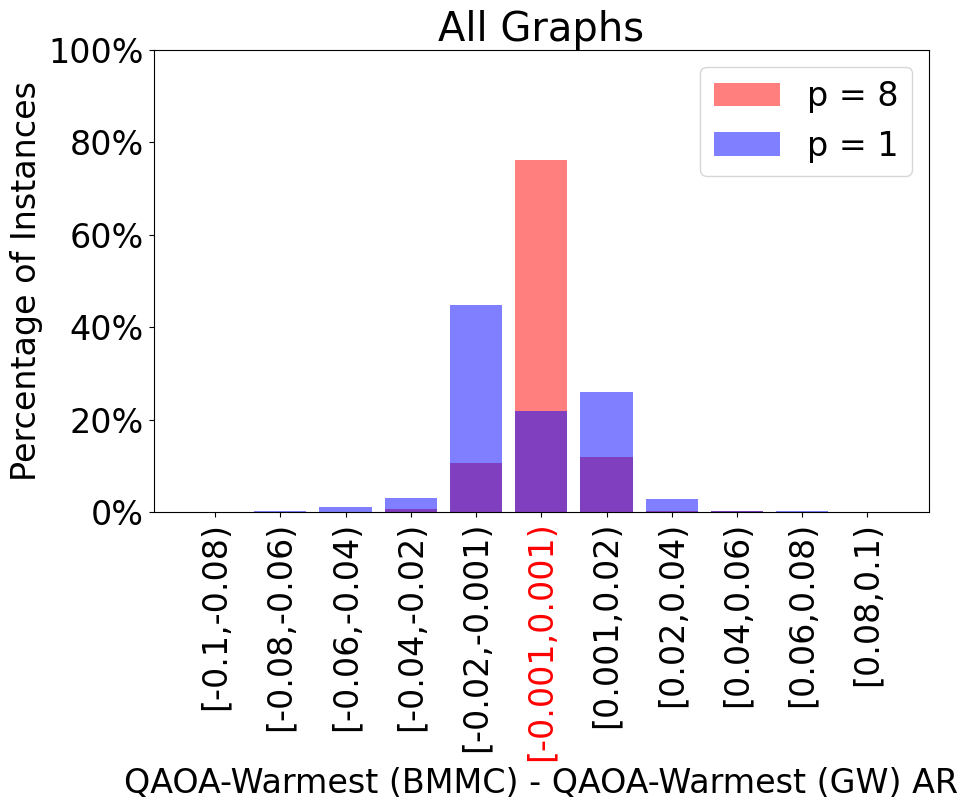}
    \caption{\footnotesize Histogram showing the difference in {(instance-specific)} approximation ratio (AR) when using QAOA-warmest (on instance library $\mathcal{G}$) with various warm-start strategies: projected GW and locally optimal BM-MC warm-starts. The blue and red bars correspond to depth $p=1$ and $p=8$ QAOA-warmest respectively with the purple regions indicating an overlap in the histograms. For both approaches, $k = 3$ solutions and vertex-at-top rotations are used to produce the figure; the results are similar if one instead uses a different combination of \dimension and/or rotation scheme.}
    \label{fig:GWVsBMMC}
\end{figure}

\subsection{GW and BM-MC Scaling}
\label{sec:GW_BMMC_Scaling}
In Section \ref{sec:warmstarts}, two methods of creating a warm-start initialization are discussed: projecting GW SDP solutions and finding approximate BM-MC$_k$ solutions. Figure \ref{fig:GW_BMMC_Scaling} uses instances from the MQLib \cite{DGS18} library\footnote{MQLib \cite{DGS18} is a diverse library of Max-Cut instances; the results of Figure \ref{fig:GW_BMMC_Scaling} may differ if different types/families of graphs are used (e.g. random Erd\H{o}s–R\'enyi graphs). For each MQLib instance, we used the exact Max-Cut solver BiqCrunch \cite{KMR17} to try to find the optimal cut. Figure \ref{fig:GW_BMMC_Scaling} uses all positive-weighted MQLib instances up to 663 nodes for which BiqCrunch was able to find the optimal cut within 24 hours. } to compare these two methods at different dimensions ($k=2,3$) with respect to the {(instance-specific)} approximation ratio they achieve with hyperplane rounding; these approximation ratios are compared against hyperplane rounding of the $n$-\dimensional GW solution. It is clear from the figure that the projecting of GW solutions preserves the approximation ratio (from hyperplane rounding); this is consistent with the results of Theorem \ref{thm: subspace-hyperplane-rounding}. On the other hand, while BM-MC$_k$ solutions preserve the approximation ratio (from hyperplane rounding) for small graphs, the gap in approximation ratios (compared to $n$-\dimensional GW hyperplane rounding) grows as the number of nodes increases.

\subsection{Interesting Instances}
In Section \ref{sec:QAOAWarmestCompare}, Table \ref{fig:pieChartTable} shows that at circuit depth $p=8$, there are five instances for which QAOA-warmest did not achieve the highest {(instance-specific)} approximation ratio compared to the other algorithms considered.

Of these five instances, standard QAOA was the best algorithm for precisely two of these (instances \#778 and \#1820). For the remaining three instances (\#1698, \#1889, \#2010), GW was the best algorithm; however, all three of these instances have the property that there is a single negative edge-weight whose magnitude is much larger than the other edge weights in the graph and additional numerical simulations show that a suitable vertex-at-top rotation (selecting the vertex that is incident to the large-magnitude negative edge weight) allows QAOA-warmest to outperform GW.

Table \ref{fig:problematicInstances} gives detailed statistics for the approximation ratios achieved by each of the Max-Cut algorithms considered for these five instances.

\subsection{Comparison with Egger et al.}
Figure \ref{fig:eggerComparison} compares the {(instance-specific)} approximation ratios achieved by QAOA-warmest and a variant of QAOA considered by Egger et al. \cite{egger2020warm}. In the context of the Max-Cut problem, Egger et al. considered an approach which takes a good starting cut $(S, V \setminus S)$ (obtained via GW or possibly other means) and uses this cut to construct an initial quantum state $\ket{s_0}$. With this modified initial quantum state and an appropriate modification of the mixing Hamiltonian, Egger et al. show that their variation of QAOA is able to recover the cut at circuit depth $p=1$, i.e., there is a choice of variational parameters $\gamma_1$ and $\beta_1$ such that the only cut obtained at those parameters is precisely $(S,V \setminus S)$.

To give a fair comparison for Egger et al.'s approach, we consider 10 cuts generated by the GW algorithm and take the best 5. Due to the size of the instances we consider, usually at least one of the best 5 cuts would be optimal and hence Egger et al.'s approach would essentially already start with an optimal solution which is not interesting. For this reason, in Figure \ref{fig:eggerComparison}, we only consider those instances (22.7\% of the instance library) for which neither QAOA-warmest or Egger et al.'s approach starts with the optimal solution.

For Egger et al.'s approach, we consider two different choices for initialization of the variational parameters: (1) near the origin and (2) the choice of parameters that recovers the value of cut used to initialize the QAOA variant (i.e. $\beta_1 = \pi/2$ with the remaining parameters being set to zero). In both cases, Figure \ref{fig:eggerComparison} demonstrates that QAOA-warmest typically has the superior performance.

There are a total of 163 instances for which the approximation ratio achieved by depth-8 QAOA-warmest beats GW (by at least 0.001) and the approximation ratio achieved by GW beats Egger et al.'s approach (with initialization  $\beta_1 = \pi/2$ with the remaining parameters being set to zero) at depth-8 (by at least 0.001). For these instances, the median gap in approximation ratio between QAOA-warmest and GW was 0.0466 and the median gap in approximation ratio between GW and Egger et al.'s approach is 0.0458.

\begin{table*}[htbp]
\centering
\begin{tabular}{|c|c|c|c|}
\hline
     Instance &  \# Nodes & \# Edges & Approx. Ratio\\\hline
     $J(6,3,1)$ & 20 & 90 & $0.9123$\\
     $J(8,4,1)$ & 70 & 560 & $0.8889$\\
     $J(10,5,1)$ & 252 & 3150 & $0.8810$\\
     $J(10,5,2)$ & 252 & 12600 & $0.9402$\\
    $J(12,6,1)$ & 924 & 16632 & $0.8787$\\
    $J(12,6,2)$ & 924 &103950 & $0.9123$\\
    \hline
\end{tabular}
\caption{\footnotesize Small ($<1000$ nodes) instances using Karloff's \cite{K99} construction. For each instance, we include the number of nodes, edges, and the theoretical expected approximation ratio one would obtain using Goemans-Williamson algorithm on that instance.}
\label{tab:karloff}
\end{table*}

\section{Twenty-Node Graph}
\label{sec:twentyNodeGraph}

In this section, we discuss why the graph considered in Figure \ref{fig:Guadalupe_device} is interesting. The graph used is a 20-node instance, where Goemans-Williamson achieves an {(instance-specific)} approximation ratio of 0.912. We briefly summarize the construction and properties of this graph. 

Recall that the worst-case approximation ratio for the Goemans-Williamson (GW) algorithm is 0.878. Karloff \cite{K99} showed that the 0.878 bound for GW is tight by constructing a family of graphs whose {(instance-specific)} approximation ratios approaches 0.878 as graph size increases. The construction for this family of instances is as follows: consider non-negative integers $b \leq t \leq m$ and let $J(m,t,b)$ denote the graph with vertex set $\binom{m}{t}$, i.e., the vertices are all $t$-element subsets of $[m]$; two distinct vertices/subsets $S$ and $T$ of $J(m,t,b)$ are adjacent if and only if they have exactly $b$ elements in common, i.e. $|S \cap T| = b$. {Karloff proved the approximation ratio for GW on some of these instances:}

\begin{theorem}{\cite{K99}}\label{Thm-Conj} Let $m$ be an even positive integer and $G = J(m,m/2,b)$. If $0 \leq b \leq m/12$, then the {(instance-specific)} approximation ratio for Goemans-Williamson on $G$ is given by $$\frac{\frac{1}{\pi}\arccos(\frac{4b}{m}-1)}{1-\frac{2b}{m}}.$$
\end{theorem}

Non-trivial instances using Karloff's construction arises once $b \geq 1$; however, in order for the hypotheses of Theorem \ref{Thm-Conj} to be satisfied, this requires $m \geq 12$ which implies that one needs to consider instances with at least ${\binom{12}{6}} = 924$ nodes. Performing reliable experiments on such large instances is not feasible for current quantum hardware.

We noticed that smaller instances with a weak {(instance-specific)} GW approximation ratio could be constructed if the following conjecture by Karloff was true. 

\begin{conjecture}\label{thm:conjecture} (Conjecture 2.12 in \cite{K99})
If $m$ is an even positive integer with $0 \leq b < m/4$, then the smallest eigenvalue of the adjacency matrix of $J(m,m/2,b)$ is ${\binom{m/2}{b}}^2\left(\frac{4b}{m}-1\right)$.
\end{conjecture}

Karloff \cite{K99} proved a special case of this conjecture in the case where $0 \leq b \leq m/12$ which was instrumental in showing Theorem \ref{Thm-Conj} (with the same inequality in the hypotheses).

 In 2018, Brouwer et al. \cite{BCIM18} proved the following theorem below; substituting $t = m/2$ into Theorem \ref{thm:brouwer} and performing a few simple calculations, they also found that Conjecture \ref{thm:conjecture} follows as a corollary.\footnote{The statement of Theorem \ref{thm:brouwer} has been modified in order to be consistent with the notation used in Karloff's construction \cite{K99}.}

\begin{theorem}
\label{thm:brouwer}
 (Theorem 3.10 in \cite{BCIM18}) Let $0 \leq b < t$ and let\footnote{The eigenvalues of the adjacency matrix of $J(m,t,b)$ are $\beta_0,\beta_1,\dots,\beta_t$ (each with positive multiplicity) \cite{knuth91}.}
 
\[\beta_s = \sum_{r=0}^s (-1)^{s-r}{\binom{s}{r}}{\binom{t-r}{t-b}}{\binom{m-t-s+r}{m-2t+b}}. \] Then $\beta_1$ is the smallest eigenvalue of the adjacency matrix of $J(m,t,b)$ if and only if $(t-b)(m-1) \geq t(m-t)$. In this case, $\beta_1$ is also the second largest in absolute value among the eigenvalues of the adjacency matrix of $J(m,t,b)$.
\end{theorem}

Conjecture \ref{thm:conjecture} allows us to relax the inequality in Theorem \ref{Thm-Conj} to $0 \leq b < m/4$; thus, it can be applied with $m=6$ and $b=1$ to obtain the graph $J(6,3,1)$ (used in Figure \ref{fig:Karlov_20}) with ${\binom{6}{6/2}} = 20$ nodes where GW achieves an {(instance-specific)} approximation ratio of
$$\frac{\frac{1}{\pi}\arccos(\frac{4}{6}-1)}{1-\frac{2}{6}} = 0.912.$$

Table \ref{tab:karloff} shows all non-trivial instances under 1000 nodes that use Karloff's \cite{K99} construction for which we can calculate the approximation ratio using Theorem \ref{Thm-Conj} and Conjecture \ref{thm:conjecture}. These instances may be of interest to those working with near-term quantum devices.

\end{document}